\PassOptionsToPackage{unicode}{hyperref}
\PassOptionsToPackage{hyphens}{url}
\PassOptionsToPackage{dvipsnames,svgnames,x11names}{xcolor}
\documentclass[12pt]{article}

\usepackage{natbib}
\usepackage{amsmath,amssymb}
\usepackage{nicefrac}
\RequirePackage{amsthm,amsfonts}
\usepackage[T1]{fontenc}
\usepackage[utf8]{inputenc}
\usepackage{lmodern}
\usepackage[]{microtype}
\UseMicrotypeSet[protrusion]{basicmath} 
\usepackage{parskip}
\usepackage{xcolor}
\usepackage{booktabs}
\usepackage{graphicx}
\usepackage{caption}
\usepackage{subcaption}

\usepackage{bookmark}
\IfFileExists{xurl.sty}{\usepackage{xurl}}{} 
\hypersetup{
  pdftitle={A bridge representation of Gaussian Whittle-Matérn fields on compact metric graphs},
  pdfauthor={David Bolin, Alexandre B. Simas, and Jonas Wallin},
  pdfkeywords={Maximum likelihood, kriging, stochastic partial differential equation, Markov random field},
  colorlinks=true,
  linkcolor={blue},
  filecolor={Maroon},
  citecolor={Blue},
  urlcolor={Blue}}

\newcommand{\proper}{\mathsf}
\newcommand{\pE}{\proper{E}}
\newcommand{\pN}{\proper{N}}
\newcommand{\mv}[1]{{\boldsymbol{\mathrm{#1}}}}
\DeclareMathOperator{\pd}{\partial}
\DeclareMathOperator{\Cov}{Cov}
\newcommand{\trsp}{\ensuremath{\top}}
\newcommand{\A}{\mathcal{C}}
\newcommand{\Ac}{\mathcal{U}}

\newtheorem{theorem}{Theorem}
\newtheorem{proposition}{Proposition}
\newtheorem{Lemma}{Lemma}
\newtheorem{Corollary}{Corollary}
\newtheorem{definition}{Definition}
\allowdisplaybreaks

\begin{document}

\title{\bf A bridge representation of Gaussian Whittle--Mat\'ern fields on compact metric graphs}
\author{David Bolin\thanks{
  The authors gratefully acknowledge support by King Abdullah University of Science and Technology (KAUST) under Award No. ORFS-CRG11-2022-5015.}, Alexandre B. Simas\hspace{.2cm}\\
  Statistics Program, King Abdullah University of Science and Technology\\
  and \\
  Jonas Wallin \\
  Department of Statistics, Lund University}
\date{}
\maketitle

\bigskip
\begin{abstract}

Gaussian Whittle--Mat\'ern fields form a flexible class of Gaussian processes on compact metric graphs, where spatial dependence is governed by the geometry and connectivity of the network through a fractional-order stochastic
partial differential equation. This paper develops a new bridge representation of these fields in the case of half-integer smoothness parameters, when the fields have Markov properties. This representation decomposes the field into a
finite-dimensional graph component and independent Whittle--Mat\'ern bridge
processes on the individual edges. The resulting decomposition leads to efficient
likelihood evaluation, kriging prediction, and
simulation methods.
We show that this improves numerical stability and can greatly reduce computation time compared to previous methods. A simulation study on a Chicago street-network graph illustrates the computational efficiency of the sampling method and an application to Madrid traffic intensity data demonstrates the practical gains for likelihood-based inference and prediction.
The methods are implemented in the \texttt{R} package \texttt{MetricGraph}.

\end{abstract}

\noindent%

\bigskip

\section{Introduction}\label{sec:introduction}

Spatial data are often observed on domains whose geometry is determined by a network. 
In hydrological applications, variables such as dissolved oxygen, stream temperature, and ecological indicators are measured along river systems, 
where spatial dependence is affected by flow connectivity, branching, and distance along the stream network \citep{isaak2014applications}. 
In urban applications, traffic counts, speeds, and intensities are observed on street networks, 
where dependence is more naturally related to movement along roads than to straight-line Euclidean distance \citep{okabe2012spatial,selby2013traffic,dahl2025modeling}. 
These examples have a common statistical feature: the observation locations lie on a one-dimensional network embedded in Euclidean space, and the geometry relevant for prediction is determined by both the lengths of the edges and the connectivity of the network.

A convenient mathematical representation of such domains is a metric graph. 
A metric graph $\Gamma$ is specified by a set of vertices $\mathcal{V}\subset \mathbb{R}^p, p \in \mathbb{N}$ and a set of edges $\mathcal{E}$ connecting the vertices,
  where each edge $e$ is defined by a rectifiable curve $\gamma:[0,\ell_e]\to \mathbb{R}^p$ of length $\ell_e>0$, and parameterized by arc-length. A location $s\in \Gamma$ is a position on some edge $e\in\mathcal{E}$.

There has recently been considerable interest in defining Gaussian random fields
on metric graphs. Stream-network models \citep{CressieRiver, Hoef2006, Hoef2010} are widely used in hydrological
applications \citep{isaak2014applications}, and valid isotropic covariance models \citep{anderes2020isotropic, porcu2022} and related computational methods for simulation \citep{alegria2026computationally} have been
developed for particular classes of metric graphs. 
An approach which is valid for any compact metric graph is to formulate a Gaussian process through a stochastic partial differential equation (SPDE) on the metric graph. Specifically,  \citet{BSW2022} introduced the Whittle--Mat\'ern fields on metric graphs as the solution to 
\begin{equation}\label{eq:Matern_spde}
	(\kappa^2 - \Delta_\Gamma)^{\alpha/2} (\tau u) = \mathcal{W}, \qquad \text{on $\Gamma$},
\end{equation}
where $\kappa,\tau>0$ are parameters, 
$\Delta_\Gamma$ is a Laplacian on the graph, $\mathcal{W}$ is Gaussian white noise, and the smoothness of the field is controlled by the parameter $\nu>0$, with $\alpha = \nu + \nicefrac{1}{2}$.  
This construction is a natural extension of the classical SPDE representation of Mat\'ern fields from Euclidean domains \citep{whittle63} to metric graphs. In $\mathbb{R}^p$, the corresponding SPDE takes the form \eqref{eq:Matern_spde} with $\alpha = \nu + \tfrac{p}{2}$ and the Laplacian replacing $\Delta_\Gamma$, and its solution is a centered Gaussian random field with Mat\'ern covariance function  
\begin{equation}\label{eq:matern_cov}
	\varrho_M(h) = \frac{\tau^{-2}}{2^{\nu-1}\Gamma\!\left(\nu + \tfrac{p}{2}\right)(4\pi)^{p/2}\kappa^{2\nu}} (\kappa |h|)^{\nu} K_\nu(\kappa |h|),
\end{equation}
where $\tau, \kappa, \nu > 0$ control the marginal variance, practical correlation range, and smoothness, $K_\nu(\cdot)$ is the modified Bessel function of the second kind, and $\Gamma(\cdot)$ is the gamma function. 

\citet{split1} introduced a method for evaluating the likelihood of the Whittle--Mat\'ern fields on metric graphs exactly when $\alpha\in\mathbb{N}$, which is based on adding all observation locations as vertices in the graph and taking advantage of the Markov properties the fields have for these values of $\alpha$ \citep{BSW_Markov}. Even though this facilitates computationally efficient inference for large graphs and datasets, there are two disadvantages of the method. First, adding the observation locations as vertices requires an additional preprocessing step that may be computationally heavy for very large graphs. Second, if one has observation locations that are very close, numerical instabilities may arise, especially for $\alpha>1$. The reason for this is that one then needs to evaluate covariances for small distances, which may cause matrices to become numerically singular. This is also a well-known problem when working with Mat\'ern covariances on Euclidean domains.

In this work, we introduce a new ``bridge'' representation of Whittle--Mat\'ern fields on compact metric graphs with  $\alpha\in\mathbb{N}$. The representation decomposes the field into a sum of independent Gaussian processes on the edges,  which are zero at the vertices, and a ``low-rank'' Gaussian process defined in the vertices and interpolated to the edges.
This makes it possible to condition locally within each edge while keeping the sparse linear algebra tied to the original graph,
rather than to one expanded by all observation or prediction locations.

We use this representation to obtain computationally efficient and numerically stable formulas for likelihood
evaluation and kriging prediction. We further develop exact finite-dimensional simulation methods based on the bridge decomposition, including a kriging-corrected edge sampler. 
We demonstrate the computational advantages of the methods through a  simulation study on a Chicago street network and an application to Madrid traffic intensity data.

The rest of the paper is organized as follows. Section~\ref{sec:construction}
reviews the metric-graph notation, the Whittle--Mat\'ern construction, and the
finite-dimensional representation used by \citet{split1}.
Section~\ref{sec:bridge} derives the bridge representation.
Sections~\ref{sec:likandpred} and~\ref{sec:predict} show how the representation can be used for likelihood evaluation, prediction, and sampling.
Section~\ref{sec:simulation} presents a simulation study using the Chicago
street-network graph, and
Section~\ref{sec:application} presents the Madrid traffic application. The
article concludes with a discussion in Section~\ref{sec:discussion}. The new
methods are implemented in the \texttt{R} package \texttt{MetricGraph}
\citep{bsw_MetricGraph_cran}.

\section[Metric graph definition and notation]{Metric graph definition and notation}\label{sec:construction}
Throughout, $\Gamma$ denotes a connected compact metric graph with vertex set $\mathcal{V}$ and edge set $\mathcal{E}$.
Each edge $e\in\mathcal{E}$ is identified with an interval $[0,\ell_e]$ through an arc-length parametrization, and we write $(\underline e,\bar e)$ for the initial and terminal vertices of $e$.
Thus a location $s\in\Gamma$ can be represented as $s=(e,t)$, where $t\in[0,\ell_e]$.
This coordinate representation is unique for points in the interior of an edge.
In contrast, a vertex $v\in\mathcal{V}$ may have several coordinate representations, corresponding to the endpoints of the edges incident to $v$.
To indicate that a coordinate representation $s=(e,t)$ represents the vertex $v$, we write $s\in v$.
We denote by $\mathcal{E}_v$ the set of edges incident to $v$, by $L_v$ the set of loops based at $v$, and define
$\deg(v)=|\mathcal{E}_v|+|L_v|.$
With this convention, each loop contributes two coordinate representations at its base vertex.
We also fix enumerations $(v_1,\ldots,v_{|\mathcal{V}|})$ and $(e_1,\ldots,e_{|\mathcal{E}|})$.

For a function $f$ on $\Gamma$, we write $f_e=f|_e$ for its restriction to the edge $e$.
If $f_e$ is differentiable at an endpoint, and if $s=(e,t)$ is the corresponding coordinate representation of a vertex, then $\partial f(s)$ denotes the directional derivative of $f$ along $e$ in the direction away from that vertex. That is, if $e$ is identified with $[0,\ell_e]$, then $\partial f(e,0)=f_e'(0)$ and $\partial f(e,\ell_e)=-f_e'(\ell_e)$.

The Kirchhoff Laplacian $\Delta_\Gamma$ is the differential operator that acts as the second derivative on each edge and imposes the usual continuity and Kirchhoff vertex conditions.
Its domain is
$\mathcal{D}(\Delta_\Gamma)=\widetilde H^2(\Gamma)\cap C(\Gamma)\cap K(\Gamma),$
where
$\widetilde H^k(\Gamma)=\bigoplus_{e\in\mathcal E} H^k(e)$ is the set of functions that are $k$ times weakly differentiable on the edges, 
$C(\Gamma)=\{f\in L_2(\Gamma): f \text{ is continuous on } \Gamma\}$,
and
$$K(\Gamma)=\left\{f\in \widetilde H^2(\Gamma)\cap C(\Gamma): \sum_{s\in v}\partial f(s)=0 \text{ for every } v\in\mathcal V\right\}.$$
Thus, functions in $\mathcal{D}(\Delta_\Gamma)$ are continuous across vertices and satisfy a sum-to-zero condition on the outgoing derivatives at each vertex.
For vertices of degree two that are not loops, this condition reduces to continuity of the derivative across the vertex.

This is the operator used in the Whittle--Mat\'ern SPDE \eqref{eq:Matern_spde}.
Specifically, a Whittle--Mat\'ern field on $\Gamma$ with parameters $(\kappa,\tau,\alpha)$ is the centered Gaussian field obtained as the solution to $L^{\alpha/2}(\tau u)=\mathcal W,$ where $\mathcal W$ denotes Gaussian white noise on $L_2(\Gamma)$ and $L=\kappa^2-\Delta_\Gamma$ is the shifted Kirchhoff Laplacian.
The fractional power is understood in the spectral sense.
For $\alpha>1/2$, this equation has a unique solution in $L_2(\Gamma)$, $\mathbb P$-almost surely, and the solution admits a continuous modification; see \citet[][Proposition~1]{BSW2022}.

In the rest of the paper, we use this construction through its Markov and finite-dimensional representations for integer $\alpha$.
In particular, for $\alpha\in\mathbb N$ we use the edge-wise vectors
$\mv{u}_e(t)=\bigl[u_e(t),u_e'(t),\ldots,u_e^{(\alpha-1)}(t)\bigr]^\top,$
where all derivatives are taken with respect to the edge coordinate $t$.
The precise spectral and weak formulation of the SPDE is only needed in the proof of our main result (Theorem~\ref{thm:ReprTheoremEdge_Refined}) and is recalled there.

Based on the parameters $(\kappa,\tau,\alpha)$ we define the variance parameter 
\begin{equation}\label{eq:tau_microergodic}
\sigma^2 = \frac{\Gamma(\alpha-1/2)}{\tau^2  \kappa^{2\alpha-1} 2\sqrt{\pi}\Gamma(\alpha)},
\end{equation}
which is the marginal variance related to \eqref{eq:matern_cov} for $p=1$, i.e., $\sigma^2 = \varrho_M(0)$. Note, however, that a Whittle--Mat\'ern field in general is non-stationary, and its marginal variances are different from $\sigma^2$ close to vertices of degree different from two.

\section{The bridge representation}\label{sec:bridge}
The definition of the Whittle--Mat\'ern fields is difficult to use for simulation and inference. \citet{split1} derived an alternative representation of these fields with $\alpha\in\mathbb{N}$ by conditioning certain independent Gaussian processes defined on the edges to satisfy the Kirchhoff vertex conditions, which facilitates exact likelihood evaluation and prediction by adding the locations of interest as vertices in the graph and using the fact that this does not change the model. 

In this section, we derive the alternative bridge representation. 
In the next section, we show that this representation can be used for likelihood-based inference and prediction without adding the locations of interest as vertices, reducing the computational cost. 
To state the bridge representation, we begin by introducing the Whittle--Mat\'ern bridge process. To distinguish between local and global objects, we use $x$ for a process defined only on a single edge and $u$ for the corresponding process on the whole graph.

\begin{definition}\label{def:WMB}
Let $x_\alpha$ be a centered Gaussian process on an interval $e = [0, \ell]$, with covariance function \eqref{eq:matern_cov} with $p=1$,  
evaluated at $h=t_1-t_2$ for $t_1,t_2\in e$, where $\nu=\alpha-\frac{1}{2}$ and $\alpha \in \mathbb{N}$.
Then, the Whittle--Mat\'ern bridge process with parameters $(\kappa,\tau,\alpha)$ on $e$, $x_{B,e}$,  
with respect to $x_\alpha$ is obtained by conditioning $x_\alpha$ on $\{\mv{x}_{\alpha}(0)  =0, \mv{x}_{\alpha}(\ell) =0\}$, that is, ${x_{B,e}(t) = x_\alpha(t) | \{\mv{x}_{\alpha}(0)  =0, \mv{x}_{\alpha}(\ell) =0\}},$ 
where $ \mv{x}_{\alpha}(t)= [x_\alpha(t),x_\alpha^{(1)}(t),\ldots , x_\alpha^{(\alpha-1)}(t)]$.
\end{definition}

This process has the properties shown in the following proposition, where the operator $B^{\alpha}$ maps a sufficiently smooth function on an interval $[0,\ell]$ to its values and $\alpha-1$ first derivatives at $0$ and $\ell$, 
$B^{\alpha} u =  \left[u(0),u^{(1)}(0),\ldots,u^{(\alpha-1)}(0),u(\ell),u^{(1)}(\ell),\ldots,u^{(\alpha-1)}(\ell) \right]^\top$. 

\begin{proposition}\label{prp:Whittle_Matern_bridge_prop}
Let $x_{B,e}(\cdot)$  be a Whittle--Mat\'ern bridge process on the interval 
${e = [0, \ell]}$. Then,
 $x_{B,e}(\cdot)$
is $\alpha-1$ times differentiable, its derivatives 
are  continuous,
and $B^{\alpha}x_{B,e}(\cdot) = \mv{0}$. Further, $x_{B,e}(\cdot)$
has covariance function 
\begin{equation}\label{eq:cov_func_whittle_matern_bridge}
	 r_{e}(t_1,t_2) = \varrho_M(t_1-t_2) - 
		\begin{bmatrix}
			\mv{r}_1(t_1,0) & \mv{r}_1(t_1,\ell) 
		\end{bmatrix}
		\begin{bmatrix}
			\mv{r}(0,0) & \mv{r}(0,\ell) \\
			\mv{r}(\ell,0) & \mv{r}(\ell,\ell)
		\end{bmatrix}^{-1}
		\begin{bmatrix}
			\mv{r}_1(t_2,0)^\top \\
			\mv{r}_1(t_2,\ell)^\top
		\end{bmatrix},
\end{equation} 		
which is well-defined, where $t_1,t_2\in e$ and $\mv{r}(s,t)$ is the matrix given by
\begin{equation}\label{eq:R_matrix_edge_repr}
	\mv{r} : \mathbb{R} \times \mathbb{R} \to \mathbb{R}^{\alpha \times \alpha}, 
	\quad \mv{r}(t_1,t_2) = \left[ \frac{\pd^{i-1}}{\pd t_1^{i-1}}\frac{\pd^{j-1}}{\pd t_2^{j-1}}\varrho_M(t_1-t_2)\right]_{i,j\in\{1,2,\ldots, \alpha\}},
\end{equation}
with $\varrho_M(\cdot)$ given in \eqref{eq:matern_cov}, and $\mv{r}_1(\cdot,\cdot)$ denotes the first row in $\mv{r}(\cdot,\cdot)$.
\end{proposition}
The matrix-valued function \eqref{eq:R_matrix_edge_repr} is the covariance function of the multivariate process $[x(t),x'(t),x''(t),\ldots, x^{(\alpha-1)}(t)]$ if $x$ 
is a centered Gaussian process on $\mathbb{R}$ with a Mat\'ern covariance function $\varrho_M(\cdot)$.

We now present the new representation for a Whittle--Mat\'ern field restricted to a single edge of $\Gamma$. 
This decomposes the field into two independent components: a bridge process and a low-rank process that is defined on the vertices and then interpolated along the edge.

\begin{theorem}\label{thm:ReprTheoremEdge_Refined}
Let $u$ be a Whittle--Mat\'ern field on a compact metric graph $\Gamma$, obtained as a solution to \eqref{eq:Matern_spde} for $\alpha\in\mathbb{N}$.  For any $e\in\mathcal{E}$, we have the following representation of $u_e$ ($u$ restricted to the edge $e$): 
$u_e(t) = u_{B,e}(t)  + \mv{S}_{e}(t) B^\alpha u_e$, $t\in [0,\ell_e],$
where 
\begin{equation}\label{eq:edge_repr_Solution}
	\mv{S}_{e}(t) = \begin{bmatrix}
		\mv{r}_1(t,0) & \mv{r}_1(t,\ell_e) 
	\end{bmatrix}
	\begin{bmatrix}
		\mv{r}(0,0) & \mv{r}(0,\ell_e) \\
		\mv{r}(\ell_e,0) & \mv{r}(\ell_e,\ell_e)
	\end{bmatrix}^{-1},
\end{equation}
$\mv{r}(s,t)$ is the matrix given by \eqref{eq:R_matrix_edge_repr},
and $\mv{r}_1(\cdot,\cdot)$ denotes the first row in $\mv{r}(\cdot,\cdot)$. 
Finally, $u_{B,e}(t)$ is a Whittle--Mat\'ern bridge process on $[0,\ell_e]$, independent of $B^\alpha u_e$. 
\end{theorem}
This representation shows that, given a Whittle--Mat\'ern field $u$ on a compact metric graph $\Gamma$, the conditional distribution of the field on any edge $e$, given the boundary data $B^\alpha u$, is independent of the rest of the graph. Moreover, this distribution depends on the geometry of the graph only through the edge length $\ell_e$.
As a corollary, we also have the following representation of the Whittle--Mat\'ern field in terms of bridge processes.

\begin{Corollary}\label{cor:bridge_representation}
Let $u$ be a Whittle--Mat\'ern field on a compact metric graph $\Gamma$, obtained as a solution to \eqref{eq:Matern_spde} for $\alpha\in\mathbb{N}$. 
Then, $u$ admits the bridge representation
\begin{equation}\label{eq:bridge_representation}
	u(s) = u_\Gamma(s)  + u_{B,e}(t), \quad s=(e,t)\in\Gamma,
\end{equation}
where  $\{ u_{B,e}\}_{e\in\mathcal{E}}$ are independent Whittle--Mat\'ern bridge processes defined on each edge. 
Further, for $s = (e,t)\in\Gamma$, we have $u_\Gamma(s) =  \mv{S}_{e}(t) \mv{D}_{e} \mv{U}$, where 
\begin{align}\label{eq:dens_u2}
   \mv{U} =  [\mv{u}_{e_1}(0)^\top,
   \mv{u}_{e_1}(\ell_{e_1})^\top,
   \mv{u}_{e_2}(0)^\top,
   \mv{u}_{e_2}(\ell_{e_2})^\top,
   \ldots,
	\mv{u}_{e_{|\mathcal{E}|}}(0)^\top,
   \mv{u}_{e_{|\mathcal{E}|}}(\ell_{e_{|\mathcal{E}|}})^\top]^\top,
\end{align}
where
 $\mv{D}_{e}$ maps $\mv{U}$ to $(\mv{u}_e(0),\, \mv{u}_e(\ell_{e}))^\top$, 
 and $\mv{S}_e(t)$ is given in \eqref{eq:edge_repr_Solution}.
\end{Corollary}

With this result, we thus have a representation of $u$ in terms of a low-rank process $u_{\Gamma}$ which is determined by the values of $u$ at the vertices, and independent bridge processes on the edges.
Figure~\ref{fig:bridge_illustration} illustrates the representation
for $\alpha=1$ on a simple metric graph, where the components are shown in
separate panels as elevations above the graph: the low-rank component
$u_\Gamma$, which interpolates the vertex values and is smooth in the
interiors of the edges, the independent bridge processes, which vanish at all
vertices, and their sum, which is a realization of the Whittle--Mat\'ern
field.

\begin{figure}[t]
	\centering
	\begin{subfigure}[t]{0.32\linewidth}
		\centering
		\includegraphics[width=\linewidth]{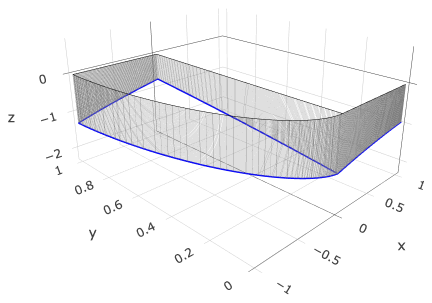}
		\caption{low-rank component $u_\Gamma$}
	\end{subfigure}
	\hfill
	\begin{subfigure}[t]{0.32\linewidth}
		\centering
		\includegraphics[width=\linewidth]{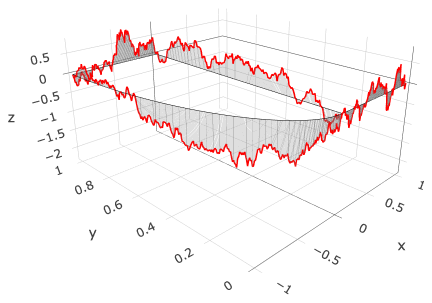}
		\caption{bridge processes $u_{B,e}$}
	\end{subfigure}
	\hfill
	\begin{subfigure}[t]{0.32\linewidth}
		\centering
		\includegraphics[width=\linewidth]{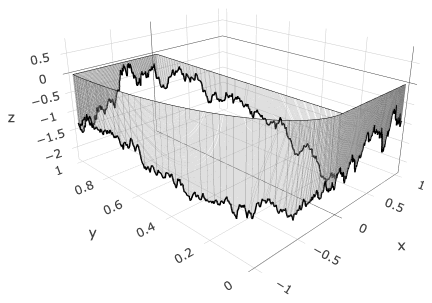}
		\caption{field $u=u_\Gamma+u_{B,e}$}
	\end{subfigure}
	\caption{Illustration of the bridge representation
	\eqref{eq:bridge_representation} for $\alpha=1$, with $\kappa=1/2$ and
	$\sigma=1$ (cf.~\eqref{eq:tau_microergodic}), on a simple metric graph: a realization of the low-rank
	component $u_\Gamma$ (a), the independent bridge processes $u_{B,e}$ (b),
	and the resulting Whittle--Mat\'ern field (c). The bridge processes vanish at the vertices, where the field is
	determined by $u_\Gamma$.}
	\label{fig:bridge_illustration}
\end{figure}

\section{Likelihood evaluation}\label{sec:likandpred}\label{sec:likelihood}
In this section we demonstrate how the bridge representation enables efficient likelihood evaluation. 
 This can offer a computational advantage compared to the approach of \citet{split1} where observation locations are added as vertices.
  However, both methods are applicable for general metric graphs, and which one is best depends on the structure of the graph and the observation locations.

Throughout what follows, $n$ and $m$ denote positive integers and $\mathcal{X}$ a non-empty set (below, $\mathcal{X}$ is an interval or $\Gamma$). Given a vector $\mv{V} \in \mathcal{X}^n$ and a set of indices $\mathcal{A}\subset\{1,\ldots,n\}$, $\mv{V}_{\mathcal{A}} = (\mv{V}_j)_{j\in\mathcal{A}}$ denotes the corresponding subvector. For a matrix $\mv{M}\in\mathbb{R}^{n\times m}$ with $(i,j)$th entry $\mv{M}_{i,j}$, and index sets $\mathcal{A}\subset\{1,\ldots,n\}$ and $\mathcal{B}\subset\{1,\ldots,m\}$, we write $\mv{M}_{\mathcal{A}\mathcal{B}} = [\mv{M}_{i,j}]_{i\in\mathcal{A},j\in\mathcal{B}}$, $\mv{M}_{\mathcal{A},} := \mv{M}_{\mathcal{A},\{1,\ldots,m\}}$, and $\mv{M}_{,\mathcal{B}} := \mv{M}_{\{1,\ldots,n\},\mathcal{B}}$. Finally, functions are applied to vectors of arguments componentwise: for $f:\mathcal{X}\to \mathbb{R}$ and $\boldsymbol{t} = (t_1,\ldots,t_n)\in \mathcal{X}^n$, $f(\boldsymbol{t}):= (f(t_1),\ldots,f(t_n))$; for $\boldsymbol{f}:\mathcal{X}\to \mathbb{R}^{m}$, $\boldsymbol{f}(\boldsymbol{t})$ is the $n\times m$ matrix with $(i,j)$th entry $\boldsymbol{f}_j(t_i)$; and for $f: \mathcal{X}\times \mathcal{X}\to \mathbb{R}$, $f(\boldsymbol{t}_1, \boldsymbol{t}_2)$ is the $n\times m$ matrix with $(i,j)$th entry $f(t_{1i}, t_{2j})$, where $\boldsymbol{t}_1 \in \mathcal{X}^n$ and $\boldsymbol{t}_2 \in \mathcal{X}^m$.

Consider observations $\mv{y} = [y_1, \dots, y_n]^\top$ collected at locations $s_1,\dots,s_n \in \Gamma$, from 
 $\mv{Y} = u(\mv{s}) + \mv{\varepsilon}$, where $\mv{\varepsilon} \sim \pN(\mv{0}, \sigma_{\varepsilon}^2 \mv{I}_n)$ is independent of $u$, $\mv{s} = (s_1,\ldots,s_n)$, and $\mv{I}_n$ stands for the $n\times n$ identity matrix. Corresponding methods for the noise-free case are presented in Appendix~\ref{app:direct}.
To exploit the bridge representation, we partition the observations according to the edges of the graph. 
For each edge $e \in \mathcal{E}$, let $n_e$ denote the number of observations located on $e$. 
We write $\mv{Y}_e$ for the subvector of $\mv{Y}$ corresponding to locations $s_i \in e$ (preserving the original ordering), 
$\mv{y}_e$ for the realized observations, 
$\mv{\varepsilon}_e$ for the corresponding noise components, 
and $\mv{t}_e$ for the vector of positions along $e$.
With $\mv{U}$ as defined in \eqref{eq:dens_u2}, the vectors $\{\mv{Y}_e\}_{e \in \mathcal{E}}$ are conditionally independent given $\mv{U}$, with
 $\mv{Y}_e \mid \mv{U} \sim \pN\left(\mv{S}_e(\mv{t}_e)\mv{D}_e\mv{U},\, \mv{\Sigma}_e\right)$. 
Here, 
$\mv{D}_e$ is defined in Corollary~\ref{cor:bridge_representation}
and 
$(\mv{\Sigma}_e)_{ij} =
\sigma_{\varepsilon}^2 \mathbb{I}(i=j) + r_{e}(t_i, t_j)$,
where  $r_{e}$ is the covariance function \eqref{eq:cov_func_whittle_matern_bridge}.
We can utilize this to create a numerically 
efficient form of the log-likelihood. We begin by considering the case ${\alpha=1}$, where $\mv{U}$ does not contain  derivatives. 

\subsection{The case $\alpha=1$}

In view of the vertex conditions, which ensure continuity of $u(\cdot)$, we have that 
$
\mv{U}_v = [u(v_1), \ldots, u(v_{|\mathcal{V}|})]^\top
$ 
is well-defined, and for a vertex $v$ one can use any of the $\deg(v)$ coordinate representations in $\Gamma$ that correspond to $v$.
We can then replace the vector $\mv{U}$ by the vector $\mv{U}_v$ in the bridge representation and use that this vector has an explicit precision matrix $\mv{Q}_v$, 
see \citet{split1}. 
The only modification to the bridge form is that we now let $\mv{D}_e$ be a map from $\mv{U}_v$ to $(\mv{u}_e(0),\, \mv{u}_e(\ell_e))^\top$. Then introducing
 the conditional expectation
 \begin{align}
	\label{eq:mu_def}
 \mv{\mu}= \mathbb{E}\left[\mv{U}_v \mid \mv{Y}=\mv{y} \right] = \widetilde{\mv{Q}}_v^{-1}\sum_{e\in\mathcal{E}_y}\mv{D}_e^\top \mv{S}_e(\mv{t}_e)^\top \mv{\Sigma}_e^{-1}\mv{y}_e,
 \end{align}
where $\widetilde{\mv{Q}}_v = \mv{Q}_v + \sum_{e\in\mathcal{E}_y}\mv{D}_e^\top \mv{S}_e(\mv{t}_e)^\top \mv{\Sigma}_e^{-1}\mv{S}_e(\mv{t}_e)\mv{D}_e$ and $\mathcal{E}_y$ is the set of edges with observations.
This gives that the log-likelihood of $\mv{y}$, up to an additive constant, is
\begin{align}
	\label{eq:lik1}
2l(\theta;\mv{y},\Gamma,\mv{s}) 
&=  \log |\mv{Q}_v| - \log|\widetilde{\mv{Q}}_v| 
- \sum_{e\in\mathcal{E}_y}\log |\mv{\Sigma}_e|  
+ \mv{\mu}^{\top} \widetilde{\mv{Q}}_v \mv{\mu} 
- \sum_{e\in\mathcal{E}_y}\mv{y}_e^{\top}\mv{\Sigma}_e^{-1}\mv{y}_e,
\end{align}
where $\theta=(\sigma_{\varepsilon},\tau,\kappa)$.
It should be noted that this formula also works for the case of direct observations (i.e., when $\sigma_{\varepsilon} = 0$), as long as one does not have observations exactly at the vertices (see Appendix~\ref{app:direct} for details of the case with observations at vertices). The advantage of this bridge formulation lies in the dimension of the matrix $\mv{Q}_v$, which has $|\mathcal{V}|$ rows instead of $|\mathcal{V}| + n$ rows in the expanded-graph approach. A drawback is that the matrices $\mv{\Sigma}_e$ are dense, leading to cubic computational complexity in evaluating their log-determinants. In the extreme case when $n \gg |\mathcal{V}|$, the expanded-graph approach has an $\mathcal{O}(n)$ cost, whereas the bridge method has an $\mathcal{O}(\sum_{e\in\mathcal{E}} n_e^3)$ cost, where $n_e$ is the number of observations on the edge $e$. Nevertheless, in practice, the number of observations on any individual edge is typically small, even if the total number of observations is large. Thus, this method often proves substantially more computationally efficient. Another advantage is that it tends to be more numerically stable, especially if one has close observations on the same edge.

\subsection{The general case $\alpha>1$}
To derive the bridge-based likelihood for $\alpha>1$, we build upon the framework by \citet{bolin2021efficient}, similarly to the approach in \citet{split1}. These methods facilitate efficient likelihood evaluation in the lower dimensional space that is obtained by applying the Kirchhoff conditions to $\mv{U}$. However, the methods are not directly applicable as they assume diagonal measurement errors, i.e., they would require $\mv{\Sigma}_e$ in the distribution of $\mv{y}_e|\mv{U}$ to be diagonal. Therefore, we now extend \citet[][Theorem 3]{bolin2021efficient} to our case.
	
Let $\mv{K}$ denote the  matrix such that $\mv{KU} = \mv{0}$ enforces Kirchhoff vertex conditions by imposing $k$ constraints (see Section 5.1 in \citealp{split1} for details). To handle these constraints, we follow \citet{split1} and introduce a change of basis $\mv{U}^{*} = \mv{T}\mv{U}$, where $\mv{T}$ is an orthogonal matrix (so that $\mv{U}=\mv{T}^\top\mv{U}^*$) constructed via a singular value decomposition of $\mv{K}$, so that the $k$ constraints from $\mv{K}$ apply solely to the first $k$ elements of $\mv{U}^{*}$, i.e., to  $\mv{U}^*_{\mathcal{C}}$, with  $\mathcal{C}=\{1,\dots,k\}$, while leaving the remaining nodes, namely $\mv{U}^{*}_{\mathcal{U}}$, unconstrained, where $\mathcal{U}=\{1,\ldots,2\alpha|\mathcal{E}|\}\setminus \mathcal{C}$. Because 
	 $\mv{Y}_{e} \mid \mv{U} = \mv{u} \sim \pN\left( \mv{S}_{e}(\mv{t}_e)\mv{D}_e\mv{u}, \mv{\Sigma}_e\right)$, we have that
	 $\mv{Y}|\mv{U} = \mv{u} \sim \pN\left(\mv{B}\mv{u}, \mv{\Sigma}\right)$,
	where $\mv{B} = [(\mv{S}_{e_1}(\mv{t}_{e_1})\mv{D}_{e_1})^\top,\ldots, (\mv{S}_{e_{|\mathcal{E}|}}(\mv{t}_{e_{|\mathcal{E}|}})\mv{D}_{e_{|\mathcal{E}|}})^\top]^\top$ is the matrix obtained by stacking the blocks $\mv{S}_{e}(\mv{t}_e)\mv{D}_{e}$ as rows and $\mv{\Sigma} = \mathrm{diag}\left(\mv{\Sigma}_{e_1},\ldots, \mv{\Sigma}_{e_{|\mathcal{E}|}}\right)$.
	
\begin{proposition}\label{Them:piAXsoft}
	Let $\mv{U}\in\mathbb{R}^m$ and $\mv{Y}\in\mathbb{R}^n$ be Gaussian random variables such that  $\mv{U} \sim \pN\left(\mv{\mu}, \mv{Q}^{-1}\right)$ and 
	$\mv{Y} | \mv{U} = \mv{u} \sim \pN\left( \mv{B}\mv{u} , \mv{\Sigma}\right)$ for some $n\times m$ matrix $\mv{B}$, where  $\mv{Q}$  and $\mv{\Sigma}$ are strictly positive-definite.  
	Further, let $\mv{K}$ be a $k \times m$ matrix of full rank and let $\mv{T}$ be the change-of-basis matrix of $\mv{K}$ discussed above.
	Define $\mv{Q}^*= \mv{T}\mv{Q}\mv{T}^\trsp$, $\mv{B}^*=\mv{B}\mv{T}^\top$, fix $\mv{b}\in \mathbb{R}^k$, and let $\mv{\mu}^*=\mv{T}\mv{\mu}$ and let $\mv{b}^*\in\mathbb{R}^k$ denote the representation of $\mv{b}$ in the new basis, i.e., the vector such that $\mv{KU}=\mv{b}$ holds if and only if $\mv{U}^*_{\mathcal{C}}=\mv{b}^*$. 
	Letting $\mathcal{C} = \{1,\ldots,k\}$ and $\mathcal{U} = \{k+1,\ldots, m\}$, the density of $\mv{Y}|\mv{KU}=\mv{b}$ is
	\begin{align*}
		\pi_{\mv{Y}|\mv{KU}}(\mv{y}|\mv{b}) = &
		\frac{|\mv{Q}^*_{\Ac\Ac}|^{\frac{1}{2}} |\mv{\Sigma}|^{-\frac12}}{\left(2\pi \right)^{\frac{n}{2}} |\widehat{\mv{Q}}^*_{\Ac\Ac} |^{\frac{1}{2}}}  
		\exp \left(-  \frac{1}{2}\left[\mv{y}^{\top}\mv{\Sigma}^{-1}\mv{y}+ \widetilde{\mv{\mu}}_{\Ac} ^{*\top} \mv{Q}^*_{\Ac\Ac}\widetilde{\mv{\mu}}^*_{\Ac}- \widehat{\mv{\mu}}_{\Ac}^{*\top}\widehat{\mv{Q}} ^*_{\Ac\Ac}  \widehat{\mv{\mu}}^*_{\Ac}\right]\right),
	\end{align*}
	where $\widehat{\mv{Q}}^*_{\Ac\Ac} = \mv{Q}^*_{\Ac\Ac}  + \left(\mv{B}^*_{,\Ac}\right)^\top \mv{\Sigma}^{-1} \mv{B}^*_{,\Ac}$, $\widehat{\mv{\mu}}^*_{\Ac} =\left(\widehat{\mv{Q}}^{*}_{\Ac\Ac} \right)^{-1}  \left(
		\mv{Q}^*_{\Ac\Ac} \widetilde{\mv{\mu}}^*_{\Ac} + \left(\mv{B}^*_{,\Ac}\right)^\top \mv{\Sigma}^{-1}\ \mv{y} \right)$, and  
	$\widetilde{\mv{\mu}}^* = 
		\left[
			(\mv{b}^*)^\top 
			(\mv{\mu}_{\Ac}^* - \left(\mv{Q}_{\Ac\Ac}^{*}\right)^{-1}  \mv{Q}^*_{\Ac\A} \left( \mv{b}^*  - \mv{\mu}^*_{\A}\right) )^\top
		\right]^\top$.
\end{proposition}

Utilizing this result allows us to write the log-likelihood, up to an additive constant, as
\begin{align}
	\label{eq:lik2}
	2l(\theta;\mv{y},\Gamma,\mv{s}) 
	&=  \log |\mv{Q}^*_{\Ac\Ac}|- \log|\widehat{\mv{Q}}^*_{\Ac\Ac} | 
	+ \widehat{\mv{\mu}}_{\Ac}^{*\top}\widehat{\mv{Q}} ^*_{\Ac\Ac}  \widehat{\mv{\mu}}^*_{\Ac}
	- \sum_{e\in\mathcal{E}_y}\log |\mv{\Sigma}_e|  
	- \sum_{e\in\mathcal{E}_y}\mv{y}_e^{\top}\mv{\Sigma}_e^{-1}\mv{y}_e,
	\end{align}
where the lower dimensional conditional expectation and precision are given by
\begin{align*}
	\widehat{\mv{Q}}^*_{\Ac\Ac} &=  \mv{T}_{\Ac,}\left( \mv{Q} + \sum_{e\in\mathcal{E}_y}\mv{D}_e^\top \mv{S}_e(\mv{t}_e)^\top \mv{\Sigma}_e^{-1}\mv{S}_e(\mv{t}_e)\mv{D}_e\right)\mv{T}_{\Ac,}^\top, \\
	\widehat{\mv{\mu}}^*_{\Ac} &=\left(\widehat{\mv{Q}}^{*}_{\Ac\Ac} \right)^{-1}  \left(
\mv{T}_{\Ac,}  \sum_{e\in\mathcal{E}_y}\mv{D}_e^\top \mv{S}_e(\mv{t}_e)^\top \mv{\Sigma}_e^{-1}\mv{y}_e \right). 
\end{align*}
These quantities can be computed solely based on sparse matrix operations and have the same computational benefits as in the $\alpha=1$ case.

The likelihood formulations for $\alpha=1$ and $\alpha>1$ also extend to
observations with covariates. For covariates $\mv{X}$ and corresponding coefficients $\mv{\beta}$, replacing $\mv{y}_e$ by
$\mv{y}_e-\mv{X}_e\mv{\beta}$ gives the corresponding likelihood. A unified
derivation for the two cases is given in Appendix~\ref{app:covariate_likelihood}.

\subsection{Likelihood evaluation with precomputed edge quantities}
\label{sec:precomputed_likelihood}

The likelihood formulas above are evaluated repeatedly during numerical
optimization. We now describe a parameter-independent precomputation that avoids
reconstructing the same edge-wise observation structure at every likelihood
evaluation.

Recall that $\theta=(\sigma_{\varepsilon},\tau,\kappa)$, and for an
edge $e\in\mathcal{E}_y$, recall that $\mv{y}_e$ denotes the observations on
the edge, and let
$\mv{t}_e=(t_{e,1},\ldots,t_{e,n_e})^\top$ denote the corresponding positions
along $e$. The parameter-independent coordinate differences needed on this edge
are
\begin{equation}
    \boldsymbol{\Delta}_{e,oo}
    =
    \left[
        t_{e,i}-t_{e,j}
    \right]_{i,j=1}^{n_e},
    \qquad
    \boldsymbol{\delta}_{e,0}
    =
    \left[
        t_{e,i}
    \right]_{i=1}^{n_e},
    \qquad
    \boldsymbol{\delta}_{e,\ell}
    =
    \left[
        \ell_e-t_{e,i}
    \right]_{i=1}^{n_e}.
    \label{eq:precomp_edge_differences}
\end{equation}
Here, $\boldsymbol{\Delta}_{e,oo}$ contains the signed differences between
interior observation locations, whereas $\boldsymbol{\delta}_{e,0}$ and
$\boldsymbol{\delta}_{e,\ell}$ contain the distances from the observations to
the left and right endpoints. These quantities are fixed during likelihood
maximization.

Using the covariance notation from Section~\ref{sec:bridge}, we write the
covariance functions with an explicit parameter argument to emphasize that the
covariance blocks are recomputed for each value of
$\theta$. Thus, the covariance matrix for
the unconditioned process at the interior observation locations is
    $\mv{C}_{e,oo}(\theta)
    =
    \varrho_M(\boldsymbol{\Delta}_{e,oo};\theta)
    :=
    \left[
        \varrho_M(t_{e,i}-t_{e,j};\theta)
    \right]_{i,j=1}^{n_e}$.
Similarly, the covariance between the interior observations and the endpoint
values and derivatives is
\begin{equation*}
    \mv{C}_{e,o\partial}(\theta)
    =
    \begin{bmatrix}
        \mv{r}_1(t_{e,1},0;\theta) & \mv{r}_1(t_{e,1},\ell_e;\theta) \\
        \vdots                     & \vdots \\
        \mv{r}_1(t_{e,n_e},0;\theta) & \mv{r}_1(t_{e,n_e},\ell_e;\theta)
    \end{bmatrix},
    \qquad
    \mv{C}_{e,\partial\partial}(\theta)
    =
    \begin{bmatrix}
        \mv{r}(0,0;\theta)        & \mv{r}(0,\ell_e;\theta) \\
        \mv{r}(\ell_e,0;\theta)   & \mv{r}(\ell_e,\ell_e;\theta)
    \end{bmatrix}.
\end{equation*}

Thus all covariance blocks needed for the bridge likelihood are obtained by
applying the covariance functions to the fixed coordinate differences in
\eqref{eq:precomp_edge_differences}, and to the fixed endpoint separation $\ell_e$. The differences are parameter
independent, whereas the numerical values of the covariance blocks change with
$\tau$ and $\kappa$.

The bridge interpolation matrix at the observation locations and the
conditional bridge covariance matrix needed for the likelihood evaluation are then
\begin{equation}
    \mv{S}_e(\mv{t}_e)
    =
    \mv{C}_{e,o\partial}(\theta)
    \mv{C}_{e,\partial\partial}^{-1}(\theta), \qquad
    \mv{\Sigma}_e
    =
    \mv{C}_{e,oo}(\theta)
    -
    \mv{C}_{e,o\partial}(\theta)
    \mv{C}_{e,\partial\partial}^{-1}(\theta)
    \mv{C}_{e,o\partial}^{\top}(\theta)
    +
    \sigma_{\varepsilon}^2\mv{I}_{n_e}.
\label{eq:precomp_bridge_matrices}
\end{equation}

The parameter-independent quantities used in the edge-wise likelihood are therefore 
\begin{equation}
    \mathcal{P}(\mv{y},\mv{s},\Gamma)
    =
    \left\{
        \left(
            \mv{y}_e,
            \boldsymbol{\Delta}_{e,oo},
            \boldsymbol{\delta}_{e,0},
            \boldsymbol{\delta}_{e,\ell},
            \ell_e,
            \mv{D}_e
        \right)
        : e\in\mathcal{E}_y
    \right\}.
    \label{eq:precomp_collection}
\end{equation}
Thus, the precomputed and non-precomputed procedures evaluate the same
likelihood. If $l$ denotes the likelihood function of either \eqref{eq:lik1} or \eqref{eq:lik2}, then 
    $l(\theta;\mv{y},\mv{s},\Gamma)
    =
    l\left(\theta;\mathcal{P}(\mv{y},\mv{s},\Gamma)\right)$.
Without precomputation, the objects in \eqref{eq:precomp_collection} are
reconstructed each time the likelihood is evaluated. With precomputation,
$\mathcal{P}$ is formed once and kept fixed during the optimization. 

\section{Prediction and sampling}\label{sec:predict}
\subsection{Prediction}
We can also utilize the bridge representation for prediction. Recall that $Y_i = u(s_i) + \varepsilon_i = u_\Gamma(s_i) +  u_{B,e_i}(t_i) + \varepsilon_i$, where $\{u_{B,e}\}_{e\in\mathcal{E}}$ are mutually independent and also independent of $u_{\Gamma}$, and $\varepsilon_i$ is independent of $u$, $i=1,\ldots,n$.
Again starting with the case $\alpha=1$, the kriging prediction is given in the following lemma:

\begin{Lemma}
\label{lem:kriging_alpha1}
Let $\Gamma$ be a compact metric graph and $u$ be a Whittle--Mat\'ern field obtained as the solution to \eqref{eq:Matern_spde}, with $\alpha=1$.
For a vector of $n$ locations $\mv{s} \in \Gamma^n,$ let $\mv{Y} = u(\mv{s}) + \mv{\varepsilon}$, where $\mv{\varepsilon} \sim \pN(\mv{0}, \sigma_{\varepsilon}^2 \mv{I}_n)$ is independent of $u$.
Then the kriging predictor of $u$  at $s^\star =(e^\star,t^\star) \in \Gamma$ is
\begin{align} \label{eq:kriging_alpha1}
\mathbb{E}\left[u(s^\star) \mid \mv{Y}= \mv{y}\right] =& \hat{u}_\Gamma(s^\star) +   r_{e^{\star}}(t^\star, \mv{t}_{e^{\star}})\mv{\Sigma}_{e^{\star}}^{-1}\left(\mv{y}_{e^{\star}} - \hat{u}_\Gamma(\mv{s}_{e^\star})\right),
\end{align} 
 where $\mv{s}_{e^\star}$ denotes the locations of the observations on the edge $e^\star$ and $\mv{t}_{e^\star}$ the corresponding positions on that edge, i.e., the $i$th element of $\mv{s}_{e^\star}$ is $(\mv{s}_{e^\star})_i = (e^\star, (\mv{t}_{e^\star})_i)$. Further, 
 $\hat{u}_\Gamma(\cdot) = \mathbb{E}[u_\Gamma(\cdot)\mid \mv{Y}=\mv{y}]$ and
$
(\mv{\Sigma}_{e^{\star}})_{ij} =  \sigma_{\varepsilon}^2 \mathbb{I}(i=j) + r_{e^{\star}}(t_i, t_j),
$
and $t_i$ is the location of the $i$th observation on edge $e^{\star}$, and $r_{e^{\star}}$ is the covariance function \eqref{eq:cov_func_whittle_matern_bridge}.
Finally for any location $s=(e,t)\in \Gamma$, we have
\begin{align}
	\label{eq:low_rank_kriging}
\hat{u}_\Gamma(s)  = \mathbb{E}[u_\Gamma(s)\mid \mv{Y} = \mv{y}]  = \mv{S}_{e}(t) \mv{D}_{e}\mv{\mu},	
\end{align}
where the functions $\mv{S}_{\cdot}(\cdot)$ and $\mv{D}_{\cdot}$ are defined in Corollary~\ref{cor:bridge_representation}, and $\mv{\mu}$ is defined in \eqref{eq:mu_def}.
\end{Lemma}
Note that the above expression also works for direct observations (i.e., $\sigma_{\varepsilon}^2 = 0$) by setting $\sigma_{\varepsilon}^2=0$ and letting $\mv{y}_{e^\star}$ denote the interior observations on edge $e^\star$.

Computing the kriging predictor can be split into three parts: computing $\mv{\mu}$, computing $\hat{u}_\Gamma(\mv{s}^\star)$ given $\mv{\mu}$ with computational 
complexity $\mathcal{O}\left( n_{s^\star} \right)$, where $n_{s^\star}$ is the number of kriging locations, and computing the bridge term given $\hat{u}_\Gamma(\mv{s}^\star)$ with complexity of $\mathcal{O}\left( \sum_{e\in \mathcal{E}} n_{e,s^\star}^3 \right)$, where $n_{e,s^\star}$ is the number of elements of $\mv{s}^\star$ in $e$.
Note that the final steps can be parallelized. This can be compared to the $\mathcal{O}\left(n^3 + n_{s^\star} n^2 \right)$ complexity of the regular kriging predictor. 

To derive the kriging predictor for $\alpha>1$, we need to obtain a numerically efficient representation of the conditional distribution.

\begin{proposition}\label{Them:piXgby}
	Let $\mv{U}\in\mathbb{R}^m$ and $\mv{Y}\in\mathbb{R}^n$ be Gaussian random variables such that  $\mv{U} \sim \pN\left(\mv{\mu}, \mv{Q}^{-1}\right)$ and 
	$\mv{Y} | \mv{U} \sim \pN\left( \mv{B}\mv{U} , \mv{\Sigma}\right)$ for some $n\times m$ matrix $\mv{B}$, where $\mv{Q}$  and $\mv{\Sigma}$ are strictly positive-definite.  
	Further, let $\mv{K}$ be a $k \times m$ matrix of full rank and let $\mv{T}$ be the change-of-basis matrix of $\mv{K}$.
	Then,  
	$\mv{U}| \left\{\mv{KU}=\mv{b},\mv{Y}=\mv{y}\right\} \sim  \pN\left(\widehat{\mv{\mu}}, \widehat{\mv{Q}}^{\dagger}\right)$
	for $\mv{b}\in \mathbb{R}^k$, where $\widehat{\mv{Q}}$ is the precision matrix of the (degenerate) conditional distribution, given by $\widehat{\mv{Q}} = \mv{T}^\top_{\Ac,}	\widehat{\mv{Q}}^*_{\Ac\Ac} \mv{T}_{\Ac,}$ and 
	$
	\widehat{\mv{\mu}} = \mv{T}^\top
			\scalebox{0.75}{$\begin{bmatrix}
					\mv{b}^* \\
					\widehat{\mv{\mu}}^*_{\Ac} 
				\end{bmatrix}$}$. Here 
	$\widehat{\mv{Q}}^*_{\Ac\Ac}$ and 
	$\widehat{\mv{\mu}}^*_{\Ac}$ are given in Proposition~\ref{Them:piAXsoft}.
\end{proposition}

Based on this result, the kriging predictor is now obtained by setting $\hat{u}_\Gamma(\mv{s}^\star) = \sum_{e\in\mathcal{E}}\mv{S}_{e}(\mv{s}^\star) \mv{D}_e \mv{T}^\top[\mv{0}^\top,(\widehat{\mv{\mu}}^*_{\Ac})^\top]^\top$ in \eqref{eq:low_rank_kriging}. 
As for $\alpha=1$, the computational cost decomposes into three parts. 
The cost of evaluating $\widehat{\mv{\mu}}$ is now higher than for $\alpha=1$ because the vector contains both the process and its derivatives, whereas the remaining terms are unchanged.

\subsection{Simulation using the bridge representation}
\label{sec:unconditional_simulation}

The bridge representation also gives an exact method for simulating the field at
a finite set of locations on the graph. For fixed $\alpha$, the edge-wise part
of the algorithm has linear complexity in the number of simulation locations on
each edge.
  For each edge $e\in\mathcal{E}$,
let
$
    0<t_{e,1}^\star<\cdots<t_{e,m_e}^\star<\ell_e.
$
The simulation is done by first drawing the low-rank component
on the graph vertices. Conditional on this draw, the edge simulations are independent
and can therefore be generated in parallel.

For $\alpha=1$, the low-rank component is obtained by drawing
$
    \mv{U}^{\mathrm{sim}}_v \sim \pN(\mv{0},\mv{Q}_v^{-1}),
$
where $\mv{Q}_v$ is the vertex precision matrix used in
Section~\ref{sec:likelihood}. For $\alpha>1$, we instead draw the endpoint
values and derivatives subject to the Kirchhoff constraints. Let $\mv{Q}$ be
the precision matrix for the unconstrained endpoint vector $\mv{U}$ in
\eqref{eq:dens_u2}, then
$
    \mv{U}^{*,\mathrm{sim}}_{\Ac}
    \sim
    \pN\!\left(\mv{0},(\mv{Q}^*_{\Ac\Ac})^{-1}\right),
$ where $ \mv{U}^{\mathrm{sim}}
    =
    \mv{T}^\trsp [
        \mv{0} ,
        \mv{U}^{*,\mathrm{sim}}_{\Ac}
    ]^\trsp$,  $\mv{Q}^*=\mv{T}\mv{Q}\mv{T}^\trsp$ and $\mv{T}$ is the change-of-basis matrix from Section~\ref{sec:likelihood}.

We next define the edge-wise covariance matrices at the simulation locations.
Let $x_e$ be a centered stationary Mat\'ern process on $[0,\ell_e]$ with
covariance function $\varrho_M$ in \eqref{eq:matern_cov}, with $p=1$ and
$\nu=\alpha-\nicefrac{1}{2}$. For $\mv{t}_e^\star
=(t_{e,1}^\star,\ldots,t_{e,m_e}^\star)^\top$, let 
$
    (\mv{\Sigma}_{M,e}^\star)_{ij}
    =
    \varrho_M(t_{e,i}^\star-t_{e,j}^\star).
$
We also write
    $\mv{\Sigma}_{e,\partial\partial}
    =
    \Cov\!\left(B^\alpha x_e\right)$,
    and 
    $\mv{\Sigma}_{e,o\partial}^\star
    =
    \Cov\!\left(x_e(\mv{t}_e^\star),
    B^\alpha x_e\right)$.
Thus
$
    \mv{S}_e(\mv{t}_e^\star)
    =
    \mv{\Sigma}_{e,o\partial}^\star
    \mv{\Sigma}_{e,\partial\partial}^{-1}.
$
The corresponding noise-free bridge covariance matrix (the matrix $\mv{\Sigma}_e$ of Section~\ref{sec:likelihood} with $\sigma_{\varepsilon}=0$)
is $
    \mv{\Sigma}_e^\star
    =
    \mv{\Sigma}_{M,e}^\star
    -
    \mv{S}_e(\mv{t}_e^\star)
    \mv{\Sigma}_{e,\partial\partial}
    \mv{S}_e(\mv{t}_e^\star)^\trsp .
$

Given $\mv{U}^{\mathrm{sim}}$, the bridge components on different edges are
independent. Hence, for each edge with $m_e>0$, we draw
    $\mv{u}^{\star,\mathrm{sim}}_e
    \mid \mv{U}^{\mathrm{sim}}
    \sim
    \pN\!\left(
        \mv{S}_e(\mv{t}_e^\star)\mv{D}_e\mv{U}^{\mathrm{sim}},
        \mv{\Sigma}_e^\star
    \right)$.
Concatenating the vectors $\mv{u}^{\star,\mathrm{sim}}_e$ over
$e\in\mathcal{E}$ gives an exact finite-dimensional simulation of the field on
the whole graph.

An alternative implementation is to simulate an unconditioned stationary
Mat\'ern process on each edge and then impose the simulated endpoint data by a
kriging correction.

\begin{proposition}\label{prop:kriging_corrected_simulation}
Fix an edge $e$ and let $x_e$ be a centered stationary Mat\'ern process on
$[0,\ell_e]$ with covariance function $\varrho_M$ in \eqref{eq:matern_cov}, with
$p=1$ and $\nu=\alpha-\nicefrac{1}{2}$. Let
  $  \mv{x}_{e,o}=x_e(\mv{t}_e^\star)$ and
$
    \mv{x}_{e,\partial}=B^\alpha x_e .
$
For $\mv{b}_e\in\mathbb{R}^{2\alpha}$, define
    $\widetilde{\mv{x}}_e(\mv{b}_e)
    =
    \mv{x}_{e,o}
    +
    \mv{S}_e(\mv{t}_e^\star)
    \left(
        \mv{b}_e-\mv{x}_{e,\partial}
    \right)$.
We then have 
    $\widetilde{\mv{x}}_e(\mv{b}_e)
    \sim
    \pN\left(
        \mv{S}_e(\mv{t}_e^\star)\mv{b}_e,
        \mv{\Sigma}_e^\star
    \right)$.
\end{proposition}

Taking $\mv{b}_e=\mv{D}_e\mv{U}^{\mathrm{sim}}$ in
Proposition~\ref{prop:kriging_corrected_simulation} gives the same conditional
edge distribution as the direct bridge simulation above. Thus, after the single
global draw of $\mv{U}^{\mathrm{sim}}$, the corrected Mat\'ern simulations can
be performed independently across all edges.

The unconditioned Mat\'ern process used in
Proposition~\ref{prop:kriging_corrected_simulation} can be simulated
sequentially on each edge. Put $t_{e,0}^\star=0$,
$t_{e,m_e+1}^\star=\ell_e$, and
$h_{e,i}=t_{e,i}^\star-t_{e,i-1}^\star$. Define the Markov state
$
    \mv{X}_{e,i}
    =
    \left[
        x_e(t_{e,i}^\star),
        x_e'(t_{e,i}^\star),
        \ldots,
        x_e^{(\alpha-1)}(t_{e,i}^\star)
    \right]^\trsp .
$
Stationarity implies that the transition matrices depend only on 
$h_{e,i}$. Specifically, draw
$\mv{X}_{e,0}\sim\pN(\mv{0},\mv{r}(0,0))$ and, for
$i=1,\ldots,m_e+1$, draw
    $\mv{X}_{e,i}\mid \mv{X}_{e,i-1}
    \sim
    \pN\!\left(
        \mv{A}(h_{e,i})\mv{X}_{e,i-1},
        \mv{\Omega}(h_{e,i})
    \right)$,
where
    $\mv{A}(h)
    =
    \mv{r}(h,0)\mv{r}(0,0)^{-1}$ and 
    $\mv{\Omega}(h)
    =
    \mv{r}(0,0)
    -
    \mv{A}(h)\mv{r}(0,h)$.
\section{Simulation study}\label{sec:simulation}

We use the street network of the University of Chicago neighborhood, considered by
\citet{baddeley2021analysing}, as a fixed
metric graph and simulate synthetic Whittle--Mat\'ern fields on it to compare
the finite-dimensional samplers. The design is in the spirit of the simulation
study of \citet{alegria2026computationally}. After conversion to a metric graph $\Gamma$, it
has $338$ vertices and $503$ edges. We write
$|\Gamma|=\sum_{e\in\mathcal{E}}\ell_e$ for the total network length.

We simulate centered Whittle--Mat\'ern fields with $\alpha\in\{1,2\}$. In both
cases, the marginal standard deviation $\sigma$, defined in \eqref{eq:tau_microergodic}, is set to $1$ and the practical
range is set to 
$0.3\times(\text{diameter of }\Gamma)=609$.
For each $\alpha$, we put
$\nu=\alpha-\nicefrac{1}{2}$,
$
    \kappa=\frac{\sqrt{8\nu}}{609}.
$
The realization
plots use $200$ equally spaced interior locations on each edge.

We compare three samplers. The \emph{Direct sampler} implements the
conditional edge construction literally: after drawing the
low-rank process, consisting of vertex values for
$\alpha=1$ and vertex values together with endpoint derivatives for
$\alpha=2$, it loops over the edges and samples the conditional Gaussian
vector on each edge directly. The \emph{Kriging-corrected edge sampler}
instead simulates an unconditioned stationary Mat\'ern process on each edge and
then imposes the simulated endpoint data using the correction in
Proposition~\ref{prop:kriging_corrected_simulation}. The
\emph{Expanded-graph sampler} inserts all requested simulation locations as
vertices and then samples from the finite-dimensional distribution on the
expanded graph, following the construction used by \citet{split1}. It is important to note that all three samplers are exact, so their only difference is computational. 

\begin{figure}[t]
    \centering
    \begin{subfigure}[t]{0.48\linewidth}
        \centering
        \includegraphics[width=0.7\linewidth]{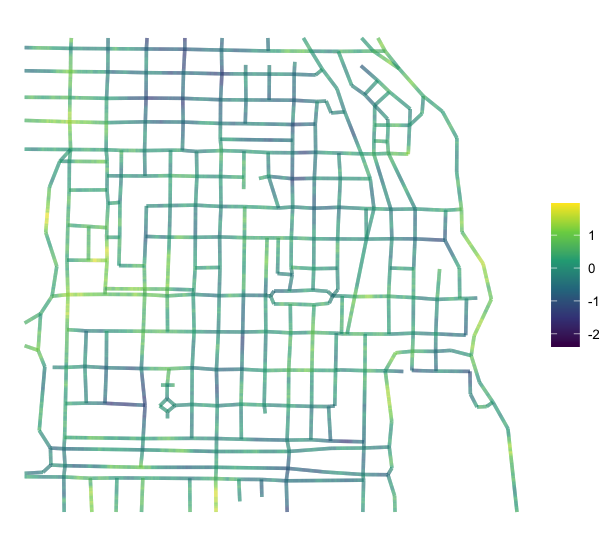}
        \caption{$\alpha=1$}
    \end{subfigure}
    \hfill
    \begin{subfigure}[t]{0.48\linewidth}
        \centering
        \includegraphics[width=0.7\linewidth]{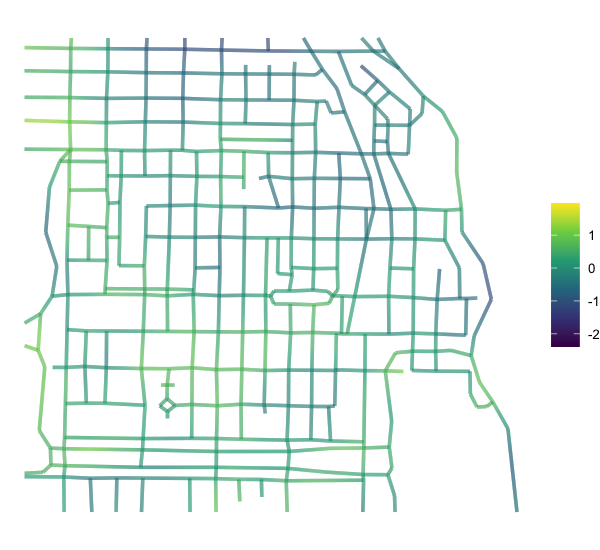}
        \caption{$\alpha=2$}
    \end{subfigure}
    \caption{Simulated Whittle--Mat\'ern fields on the Chicago street network.}
    \label{fig:chicago_sim_fields}
\end{figure}

Figure~\ref{fig:chicago_sim_fields} shows one realization for each value of
$\alpha$, generated using the kriging-corrected edge sampler. The case
$\alpha=2$ produces smoother variation along the network than the case
$\alpha=1$, as expected from the larger regularity parameter.
The timing data are reported in Table~\ref{tab:chicago_sim_data}.
The timings were obtained on a MacBook Pro with an Apple M5 Pro processor and 48 GB of RAM, without explicit parallelization. 
For the direct
edge-loop and kriging-corrected edge samplers, each entry is the median
wall-clock time over five repetitions. The kriging-corrected edge sampler is the fastest
sampler in every reported setting. The direct edge-loop sampler is competitive
on coarse grids, but its cost increases rapidly as the number of locations per
edge grows.
For $\alpha=2$, the expanded-graph sampler is substantially slower and becomes
numerically unstable on the larger grids.

For comparison, we include the spectral method of
\citet{alegria2026computationally}, which was the fastest method reported there
for simulating isotropic exponential random fields on this graph. 
Although the random fields of \citet{alegria2026computationally} are not directly
comparable to those considered here, the computational contrast is striking: by
exploiting the Markov property together with the bridge sampler, even our largest
kriging-corrected simulations are more than four times faster than the smallest approximate simulation of the isotropic model, and about two orders of magnitude faster than its smallest exact simulation.
Also note that our methods are exact, whereas the spectral method uses an approximate Gaussian density.

\begin{table}[t]
    \centering
    \scriptsize
    \setlength{\tabcolsep}{3pt}
    \begin{tabular}{llrrrrrrrr}
        \toprule
        & & \multicolumn{8}{c}{Number of simulation locations} \\
        \cmidrule(lr){3-10}
        $\alpha$ & Sampler
        & $4{,}024$ & $8{,}048$ & $16{,}096$ & $32{,}192$ & $64{,}384$
        & $128{,}768$ & $257{,}536$ & $515{,}072$ \\
        \cmidrule(lr){1-10}
        $\alpha=1$ & Direct 
        & 0.01 & 0.01 & 0.01 & 0.02 & 0.05 & 0.21 & 1.30 & 7.76 \\
        & Kriging-corrected edge
        & \textbf{0.01} & \textbf{0.01}
        & \textbf{0.01} & \textbf{0.01} & \textbf{0.02}
        & \textbf{0.03} & \textbf{0.05} & \textbf{0.10} \\
        & Expanded graph
        & 0.05 & 0.07 & 0.10 & 0.20 & 0.38 & 0.84 & 2.14 & 4.15 \\
        \cmidrule(lr){1-10}
        $\alpha=2$ & Direct 
        & 0.01 & 0.01 & 0.01 & 0.02 & 0.06 & 0.26 & 2.47 & 22.01\\
        & Kriging-corrected edge
        & \textbf{0.01} & \textbf{0.01}
        & \textbf{0.01} & \textbf{0.01} & \textbf{0.02}
        & \textbf{0.03} & \textbf{0.06} & \textbf{0.11} \\
        & Expanded graph
        & 0.05 & 0.07 & 0.11 & 0.24 & 0.56 & 1.26
        & \textemdash & \textemdash \\
        \cmidrule(lr){1-10}
        & Isotropic exp
        & 12.20 & 77.39 & 602.60 & \textemdash
        & \textemdash & \textemdash & \textemdash & \textemdash \\
        & Isotropic exp (approximate)
        & 0.52 & 0.50 & 0.86 & 1.23 & 2.11 & 4.07 & 8.39 & 17.68 \\
        \bottomrule
    \end{tabular}
    \caption{Timings in seconds for simulating one finite-dimensional
    realization on the Chicago street network.
    The number of simulation locations is the total number of interior edge
    locations, corresponding to $8$, $16$, $32$, $64$, $128$, $256$, $512$, and $1024$
    locations per edge, excluding graph vertices.
    All times are the median over five
    repetitions.
    The fastest entry within each $\alpha$ block and column is shown in bold. Dashes indicate cases where the
    implementation returned non-positive-definite matrices or was too slow to run (for Isotropic exp).}
    \label{tab:chicago_sim_data}
\end{table}

\section{Application}\label{sec:application}

To demonstrate the computational advantages of the bridge methods further, we analyze traffic intensity data sourced from the Open data portal of the Madrid City Council  (\url{https://datos.madrid.es/}), which provides real-time traffic data for the city of Madrid since 2013.  We explore the intensity of traffic averaged over the time period 6pm to 7pm from March 3 to March 7, 2025.
 The data are collected from $4{,}139$ sensors located on the streets of the city.
 The street graph is taken from OpenStreetMap (OSM) and consists of $30{,}111$ edges and $20{,}713$ vertices after pruning vertices of degree two.
The graph and observation locations are shown in Figure~\ref{figs:realdata}.

\begin{figure}[t]
	\centering
	\includegraphics[width=0.7\linewidth]{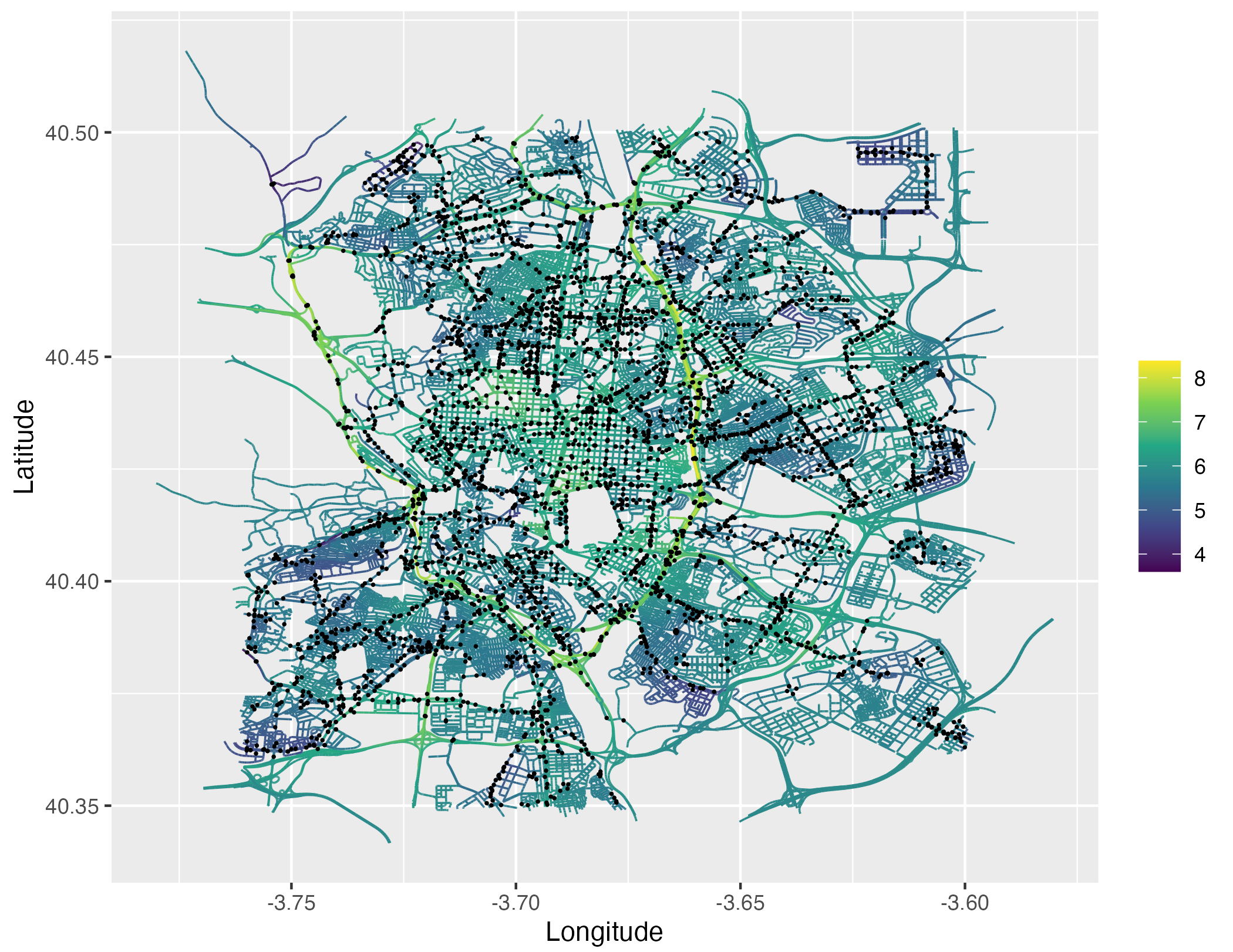}
	\vspace{-0.5cm}
	\caption{The street graph of Madrid, with the observation locations shown as black dots, and  interpolation of the log-intensity using the $\alpha=2$ Whittle--Mat\'ern model.}
	\label{figs:realdata}
\end{figure}

We first compare the computational time required to evaluate the likelihood functions for the Whittle--Mat\'ern models with $\alpha = 1, 2$, with and without the bridge method, and the isotropic exponential model of \citet{anderes2020isotropic} as a benchmark. However, the isotropic model is not guaranteed to be valid as the graph does not have Euclidean edges.

The computation times are divided into two steps: a precomputation step, which is performed only once, and the likelihood evaluation step, which is repeated during the numerical optimization of the likelihood. 
For the Whittle--Mat\'ern models, the precomputation step involves storing the data for each edge, while for the bridge method, we compute the distances between observations on the edges. For the isotropic model, the precomputation step consists of computing the resistance metric between the observation locations. This one-time cost dominates the reported precomputation time for the isotropic model. Given these distances, the remaining precomputation takes only about $1.2\times 10^{-4}$ s on average, so the direct and precomputed likelihood evaluations have essentially the same cost.
\begin{table}[t]
\centering
\small
\begin{tabular}{@{}lccc@{}}
\toprule
\textbf{Model}
& \textbf{Precompute}
& \multicolumn{2}{c}{\textbf{Likelihood}} \\
\cmidrule(lr){3-4}
& \textbf{time}
& \textbf{Precomputed}
& \textbf{Direct} \\
\midrule
WM $\alpha=1$
& 0.05 & 0.38 & 0.43 \\

WM $\alpha=1$ bridge
& 0.04 & 0.30 & 0.36 \\

WM $\alpha=2$
& 0.05 & \textemdash & \textemdash \\

WM $\alpha=2$ bridge
& 0.04 & 0.54 & 0.63 \\

Isotropic exponential
& 283.19 & 8.71 & \textemdash \\
\bottomrule
\end{tabular}
\caption{Average computation times, in seconds, for precomputation and
likelihood evaluation. WM denotes the Whittle--Matérn model.
Times are averaged over 100 repetitions, except for the isotropic
resistance-distance precomputation, which is averaged over five repetitions.
Dashes indicate unavailable timings or numerically unstable cases.}
\label{tab:timings_realdata}
\end{table}

 The resulting timings are shown in Table~\ref{tab:timings_realdata}. Due to numerical instability, the optimization for $\alpha=2$ fails unless the bridge method is used. 
 The instability is due to very short edges in the graph obtained by adding measurement locations as vertices.  For example, the locations of traffic cameras at intersections are very close to the intersection vertices (often less than 2\,m apart). In any case,
 the increased numerical stability is another advantage of the bridge method. 
It should be noted that the precomputation time for the isotropic exponential model is very large due 
to the step where the resistance metric is computed. 
The computation time after precomputation is also high due to the need to compute Cholesky factors of high-dimensional dense matrices. 
The Whittle--Mat\'ern models are substantially faster, and the bridge method reduces the precomputed likelihood-evaluation time by approximately $20\%$ for $\alpha=1$.
Finally, using precomputed quantities,
maximum-likelihood fitting of the bridge Whittle--Mat\'ern model 
took $86.4$ and $127.7$ seconds for $\alpha=1$ and $\alpha=2$,
 respectively, whereas fitting the isotropic exponential model required $68$ minutes.

The log-likelihood values for the fitted models are $-5646.46$ and $-5626.62$ for the Whittle--Mat\'ern models with $\alpha=1$ and $\alpha=2$, respectively, and $-5654.85$ for the isotropic exponential model. Thus, the model with $\alpha=2$ fits best according to the likelihood value, and since all models have the same number of parameters, also fits best according to metrics such as AIC and BIC.

Finally, to better understand the scaling of the different methods in terms of likelihood evaluation, 
we performed a simulation study where we randomly placed $n$ observation locations on the network and measured the computational time required to evaluate the log-likelihood.
 Each measurement was repeated five times to obtain average timings. 
 The results, summarized in Figure~\ref{fig:timings}, demonstrate the efficiency gains achieved by leveraging the sparsity and the bridge representation.
All experiments were run on a MacBook Pro with an Apple M5 Pro processor and 48 GB of RAM.  As shown in Figure~\ref{fig:timings}, it was not possible to run the isotropic exponential model for all values of $n$ due to insufficient RAM.

\begin{figure}[t]
	\centering
	\includegraphics[width=\linewidth]{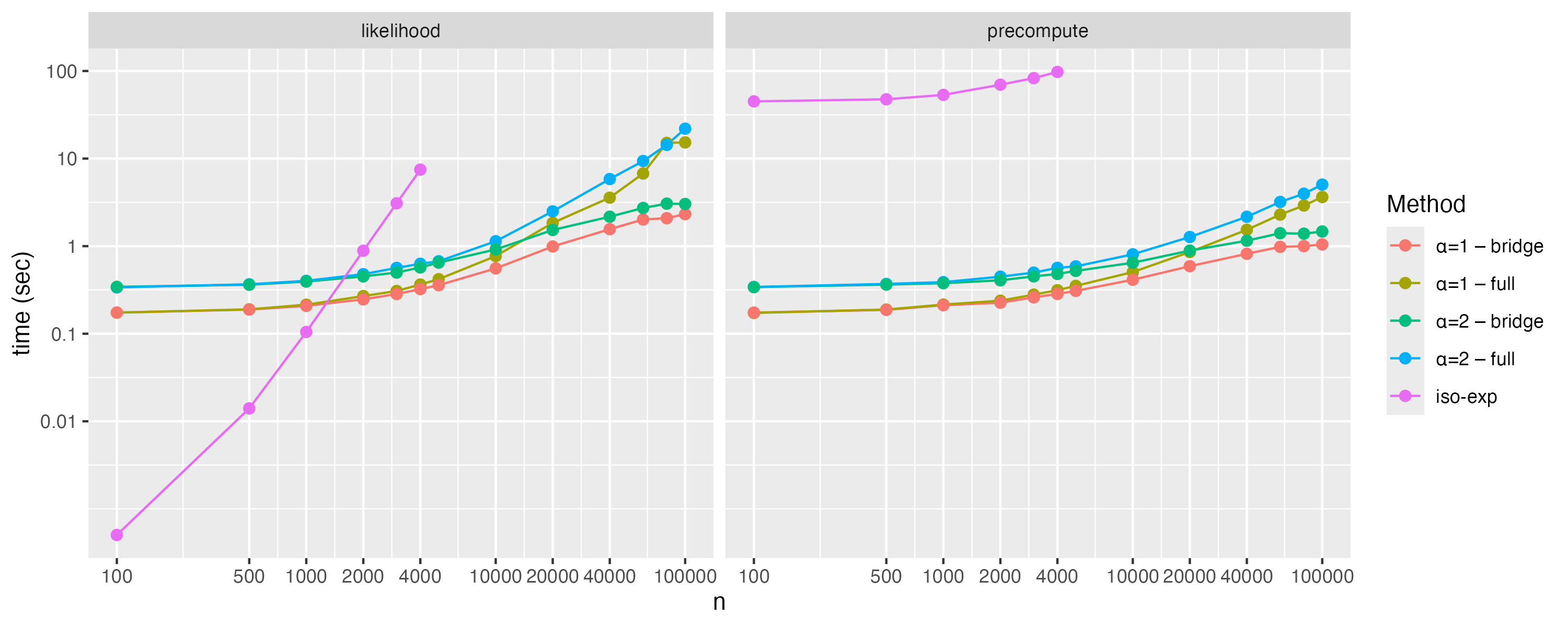}
	\vspace{-0.5cm}
	\caption{Average computation time for log-likelihood evaluation split into likelihood evaluation cost (left) and precompute cost (right), using five methods, as functions of the number of observations ($n$), based on the street network shown in Figure~\ref{figs:realdata}.}
	\label{fig:timings}
\end{figure}

\section{Discussion}\label{sec:discussion}
We have derived a new bridge representation of the Whittle--Mat\'ern fields, which allows for more computationally efficient and exact likelihood evaluation and spatial prediction for the Markov cases $\alpha\in\mathbb{N}$. Importantly, these methods are also more numerically stable than the original methods of \citet{split1}, in particular for $\alpha>1$. A limitation of the bridge approach is that the cost of the edge-wise computations is cubic in the number of observations on each individual edge, so the expanded-graph approach may remain preferable for graphs where a large number of observations concentrate on a few edges.

There are several directions in which this work can be extended. 
In particular, there are many extensions of the Whittle--Mat\'ern fields that could be considered, of which spatio-temporal models are an important future direction.  
Another interesting extension is to consider generalized Whittle--Mat\'ern fields \citep{bolinetal_fem_graph}, 
where we can allow the parameter $\kappa$ to be non-stationary over the graph. 
Another important direction is the extension to fractional orders $\alpha\notin\mathbb{N}$, for which the fields are no longer Markov and the bridge representation does not apply directly; combining the ideas presented here with the finite element and rational approximations of \citet{bolinetal_fem_graph} is an interesting topic for future research.

The code replicating all results is available at \url{https://github.com/davidbolin/Bridge}.

\appendix

\section{A unified likelihood formulation with covariates}
\label{app:covariate_likelihood}

This appendix presents a common formulation of the fixed-effect extension for
the two bridge likelihoods in Section~\ref{sec:likelihood}. The formulation below
recovers the $\alpha=1$ likelihood in \eqref{eq:lik1} and the $\alpha>1$
likelihood in \eqref{eq:lik2} through different choices of the reduced latent
vector and its precision matrix.

Consider the observation model
    $\mv{Y}=\mv{X}\mv{\beta}+u(\mv{s})+\mv{\varepsilon}$,
where $\mv{X}$ is the $n\times q$ design matrix. For each
$e\in\mathcal{E}_y$, let $\mv{X}_e$ contain the rows of $\mv{X}$
corresponding to the observations on edge $e$. Both bridge likelihoods can be
expressed in terms of a reduced latent vector $\mv{Z}$ with precision
$\mv{Q}_{0,\theta}$ and edge-wise observation matrices
$\mv{L}_{e,\theta}$ as
$\mv{Y}_e\mid\mv{Z},\mv{\beta} \sim
    \pN\left(
        \mv{X}_e\mv{\beta}+\mv{L}_{e,\theta}\mv{Z},
        \mv{\Sigma}_e \right)$,
    with 
    $\mv{Z}\sim\pN(\mv{0},\mv{Q}_{0,\theta}^{-1})$.
The two cases are obtained from
$\mv{Z} = \mv{U}_v$, $\mv{Q}_{0,\theta} = \mv{Q}_v$ and $\mv{L}_{e,\theta} = \mv{S}_e(\mv{t}_e)\mv{D}_e$ for $\alpha=1$ and 
$\mv{Z} = \mv{U}^*_{\Ac}$,
$\mv{Q}_{0,\theta} = \mv{Q}^*_{\Ac\Ac}$,
$\mv{L}_{e,\theta} = \mv{S}_e(\mv{t}_e)\mv{D}_e\mv{T}_{\Ac,}^{\top}$ for $\alpha > 1$. 

For fixed covariance parameters
$\theta=(\sigma_\varepsilon,\tau,\kappa)$, define
    $\widetilde{\mv{Q}}_\theta
    =
    \mv{Q}_{0,\theta}
    +
    \sum_{e\in\mathcal{E}_y}
    \mv{L}_{e,\theta}^{\top}\mv{\Sigma}_e^{-1}
    \mv{L}_{e,\theta}$,
    $\mv{g}_\theta
    =
    \sum_{e\in\mathcal{E}_y}
    \mv{L}_{e,\theta}^{\top}\mv{\Sigma}_e^{-1}\mv{y}_e$,
    and 
    $\mv{G}_\theta
    =
    \sum_{e\in\mathcal{E}_y}
    \mv{L}_{e,\theta}^{\top}\mv{\Sigma}_e^{-1}\mv{X}_e$.
Writing
$\mv{r}_e(\mv{\beta})=\mv{y}_e-\mv{X}_e\mv{\beta}$ and
$\mv{g}_\theta(\mv{\beta})=\mv{g}_\theta-\mv{G}_\theta\mv{\beta}$, the common
log-likelihood expression is, up to an additive constant,
\begin{align}
    2l(\theta,\mv{\beta};\mv{y},\Gamma,\mv{s})
    &=
    \log|\mv{Q}_{0,\theta}|
    -\log|\widetilde{\mv{Q}}_\theta|
    -\sum_{e\in\mathcal{E}_y}\log|\mv{\Sigma}_e|
    \nonumber\\
    &\quad
    -\sum_{e\in\mathcal{E}_y}
    \mv{r}_e(\mv{\beta})^\top\mv{\Sigma}_e^{-1}\mv{r}_e(\mv{\beta})
    +\mv{g}_\theta(\mv{\beta})^\top
    \widetilde{\mv{Q}}_\theta^{-1}\mv{g}_\theta(\mv{\beta}).
    \label{eq:unified_covariate_likelihood}
\end{align}

\section{Likelihood evaluation with direct observations}\label{app:direct}
This appendix gives the bridge likelihood for the noise-free case $\sigma_{\varepsilon}^2=0$ when exact observations may occur at graph vertices. The main likelihood formula \eqref{eq:lik1} applies directly when all exact observations are in edge interiors, so the extra bookkeeping below is only needed for observed vertices.

Let $\alpha=1$ and let $\mv{v}_{\text{obs}}$ be the set of vertices with exact observations and $\mv{v}_{\text{nobs}}$ the remaining vertices. 
Write the stacked vertex field as
\[
\mv{U}_v=\begin{bmatrix}\mv{U}_{\mv{v}_{\text{obs}}}\\ \mv{U}_{\mv{v}_{\text{nobs}}}\end{bmatrix},
\qquad
\mv{Q}_v =
\begin{bmatrix}
\mv{Q}_{oo} & \mv{Q}_{on}\\[2pt]
\mv{Q}_{no} & \mv{Q}_{nn}
\end{bmatrix},
\]
where the precision $\mv{Q}_v$ is partitioned into the observed (o) and unobserved (n) vertex indices.
By Gaussian conditioning,
$
\mv{U}_{\mv{v}_{\text{obs}}}\sim\pN\!\left(\mv{0},\,\mv{Q}_{\text{obs}}^{-1}\right),
$
with 
$\mv{Q}_{\text{obs}}:=\mv{Q}_{oo}-\mv{Q}_{on}\mv{Q}_{nn}^{-1}\mv{Q}_{no}$,
and
\[
\mv{U}_{\mv{v}_{\text{nobs}}}\mid \mv{U}_{\mv{v}_{\text{obs}}}=\mv{y}_{\text{obs}}
\sim
\pN\!\left(-\mv{Q}_{nn}^{-1}\mv{Q}_{no}\,\mv{y}_{\text{obs}},\ \mv{Q}_{nn}^{-1}\right).
\]

Let $\mv{C}_{\text{obs}}$ and $\mv{C}_{\text{nobs}}$ be the selection matrices such that
$\mv{U}_{\mv{v}_{\text{obs}}}=\mv{C}_{\text{obs}}\mv{U}_v$ and
$\mv{U}_{\mv{v}_{\text{nobs}}}=\mv{C}_{\text{nobs}}\mv{U}_v$.
Let $\mv{C}^{\text{obs}}$ and $\mv{C}^{\text{nobs}}$ be the corresponding embedding matrices that map vectors on the observed/unobserved subspaces back to the full vertex space (filling non-selected entries with zeros), so that
$\mv{C}_{\text{obs}}\mv{C}^{\text{obs}}=\mv{I}$,
$\mv{C}_{\text{nobs}}\mv{C}^{\text{nobs}}=\mv{I}$,
and $\mv{C}^{\text{obs}}\mv{C}_{\text{obs}}+\mv{C}^{\text{nobs}}\mv{C}_{\text{nobs}}=\mv{I}$.

For each edge $e$, define the $2\times 2$ diagonal matrix 
\[
\mv{A}_e := \mathrm{diag}\!\Big(\mathbb{I}\!\big((e,0)\in\mv{v}_{\text{obs}}\big),\ \mathbb{I}\!\big((e,\ell_e)\in\mv{v}_{\text{obs}}\big)\Big),
\]
that has a $1$ at the first diagonal entry if the vertex at the start of the edge is observed, and similarly a $1$ on the second diagonal entry if the vertex at the end of the edge is observed, 
and set $\mv{A}_e^{C}:=\mv{I}-\mv{A}_e$.
Let $\mathcal{E}_y$ be the set of edges with interior observations and let $\mv{y}_e$ be the vector of interior observations on edge $e\in\mathcal{E}_y$. 
Removing the contribution of \emph{observed} endpoints gives the residual
$
\tilde{\mv{y}}_e
:=
\mv{y}_e - \mv{S}_e(\mv{t}_e)\mv{A}_e\,\mv{D}_e\mv{C}^{\text{obs}}\mv{y}_{\text{obs}}.
$

Up to an additive constant, the log-likelihood can be written as
\begin{align*}
2l(\tau,\kappa;\mv{y})
=&
\log\lvert \mv{Q}_{\text{obs}}\rvert
+ \log\lvert \mv{Q}_{nn}\rvert
- \log\lvert \widetilde{\mv{Q}}_{\text{nobs}}\rvert
- \sum_{e\in\mathcal{E}_y}\log\lvert \mv{\Sigma}_e\rvert
\\
& - \mv{y}_{\text{obs}}^{\!\top}\mv{Q}_{\text{obs}}\mv{y}_{\text{obs}}
+ \mv{\mu}_{\text{nobs}}^{\!\top}\mv{Q}_{nn}\mv{\mu}_{\text{nobs}}
- \sum_{e\in\mathcal{E}_y}\tilde{\mv{y}}_e^{\!\top}\mv{\Sigma}_e^{-1}\tilde{\mv{y}}_e,
\end{align*}
where 
\begin{align*}
\mv{\mu}_{\text{nobs}}
&=
\widetilde{\mv{Q}}_{\text{nobs}}^{-1}\!
\Bigl(
-\mv{Q}_{no}\mv{y}_{\text{obs}}
+
\mv{C}_{\text{nobs}}
\sum_{e\in\mathcal{E}_y}
\mv{D}_e^{\!\top}\mv{A}_e^{C}\mv{S}_e(\mv{t}_e)^{\!\top}\mv{\Sigma}_e^{-1}\tilde{\mv{y}}_e
\Bigr),
\\
\widetilde{\mv{Q}}_{\text{nobs}}
&=
\mv{Q}_{nn}
+
\mv{C}_{\text{nobs}}
\Bigl(
\sum_{e\in\mathcal{E}_y}
\mv{D}_e^{\!\top}\mv{A}_e^{C}\mv{S}_e(\mv{t}_e)^{\!\top}
\mv{\Sigma}_e^{-1}
\mv{S}_e(\mv{t}_e)\mv{A}_e^{C}\mv{D}_e
\Bigr)\mv{C}_{\text{nobs}}^{\!\top}.
\end{align*}

Finally, the conditional mean of the full vertex vector is 
$
\mv{\mu}_v
=
\mv{C}^{\text{obs}}\mv{y}_{\text{obs}}
+
\mv{C}^{\text{nobs}}\mv{\mu}_{\text{nobs}}.
$

\section{Proofs}\label{sec:AppendixProof}

\subsection{Proofs for Section~\ref{sec:bridge}}\label{app:proofs:bridge}
We first collect the proofs supporting the bridge representation. The proof of Theorem~\ref{thm:ReprTheoremEdge_Refined} uses the edge-local Sobolev and Cameron--Martin space identifications recalled below.

\begin{proof}[Proof of Proposition~\ref{prp:Whittle_Matern_bridge_prop}]
	The expression for the covariance function \eqref{eq:cov_func_whittle_matern_bridge} follows directly from Definition~\ref{def:WMB}
	 and conditioning of normal random vectors.
	From \citet[][Theorem 2]{split1}, 
	we have that a Gaussian process with a Mat\'ern covariance function on $e$ can be written as a Whittle--Mat\'ern bridge process plus a correction term which is infinitely differentiable. The remaining properties of the Whittle--Mat\'ern bridge process are, thus, direct consequences of this fact and that a Gaussian process with a Mat\'ern covariance has these properties.
\end{proof}

Let $L_2(\Gamma)=\bigoplus_{e\in\mathcal{E}}L_2(e)$, with norm
$\|f\|_{L_2(\Gamma)}^2=\sum_{e\in\mathcal{E}}\|f_e\|_{L_2(e)}^2$. 
For $\kappa>0$, let $L=\kappa^2-\Delta_\Gamma$.
Since $\Gamma$ is compact, $L$ is self-adjoint, strictly positive-definite, and has an orthonormal basis of eigenfunctions $(\varphi_j)_{j\in\mathbb{N}}$ with eigenvalues $(\lambda_j)_{j\in\mathbb{N}}$. 
For $r>0$, the fractional power of $L$ is defined spectrally by
    $L^r \phi
    =
    \sum_{j\in\mathbb{N}}
    \lambda_j^r(\phi,\varphi_j)_{L_2(\Gamma)}\varphi_j$,
on its natural domain. 
We write
\[
    \dot{H}^{r}(\Gamma)
    =
    \left\{
        \phi\in L_2(\Gamma):
        \sum_{j\in\mathbb{N}}\lambda_j^{r}
        (\phi,\varphi_j)_{L_2(\Gamma)}^2<\infty
    \right\},
    \qquad r>0,
\]
so that $\dot{H}^{r}(\Gamma)=\mathcal{D}(L^{r/2})$ with the corresponding graph norm.

A Gaussian white noise on $L_2(\Gamma)$ is an isonormal Gaussian process
$\{\mathcal{W}(h):h\in L_2(\Gamma)\}$ satisfying
    $\pE[\mathcal{W}(h)\mathcal{W}(g)]
    =
    (h,g)_{L_2(\Gamma)}$ for all  $h,g\in L_2(\Gamma)$.
The weak solution of \eqref{eq:Matern_spde} is the centered Gaussian field satisfying
\[
    (u,\psi)_{L_2(\Gamma)}
    =
    \mathcal{W}\bigl(\tau^{-1}L^{-\alpha/2}\psi\bigr),
    \qquad \psi\in L_2(\Gamma),
\]
which for $\alpha>1/2$ has a continuous version \citep[][Proposition~1]{BSW2022}.

It remains to specify the edge-local spaces used to identify the bridge component. 
For an edge $e$ identified with $[0,\ell_e]$, let $C_c^\infty(e)$ be the set of infinitely differentiable functions with support compactly contained in the interior of $e$, and let $H_0^\alpha(e)$ be its completion in the classical Sobolev norm $\|\cdot\|_{H^\alpha(e)}$. 
On test functions, define
\[
    (u,v)_{\alpha,e}
    =
    (u,L_e^\alpha v)_{L_2(e)},
    \qquad
    L_e=\kappa^2-\frac{d^2}{dt^2},
    \qquad
    u,v\in C_c^\infty(e),
\]
and use the same notation for the continuous extension of this bilinear form to $H_0^\alpha(e)\times H_0^\alpha(e)$.

We also use the spectral edge space inherited from the graph. 
Let
    $\dot{H}^{\alpha}(e)
    =
    \{h|_e:h\in\dot{H}^{\alpha}(\Gamma)\}$,
    and
    $\dot{H}_0^\alpha(e)
    =
    C_c(e)\cap \dot{H}^{\alpha}(e)$,
where $C_c(e)$ denotes continuous functions on $e$, identified with their zero extension to $\Gamma$, whose support is compactly contained in the interior of $e$. 
For Hilbert spaces $E$ and $F$, we write $E\hookrightarrow F$ if $E\subset F$ and the inclusion map is continuous, and write $E\cong F$ if both $E\hookrightarrow F$ and $F\hookrightarrow E$ hold.

\begin{proof}[Proof of Theorem~\ref{thm:ReprTheoremEdge_Refined}]
Fix an edge $e\in\mathcal{E}$. 
By \citet[Theorem~6]{BSW_Markov}, the restriction of the Whittle--Mat\'ern field to $e$ admits an edge representation of the form
\[
    u_e(t)=\widetilde{u}_{B,e}(t)+\widetilde{\mv{S}}_e(t)B^\alpha u_e,
    \qquad t\in[0,\ell_e],
\]
where $\widetilde{u}_{B,e}$ is independent of the endpoint data $B^\alpha u_e$. 
Moreover, \citet[Lemma~1]{split1} gives the explicit interpolation matrix, which is precisely $\widetilde{\mv{S}}_e(t)=\mv{S}_e(t)$ with $\mv{S}_e(t)$ as in \eqref{eq:edge_repr_Solution}.
It remains to identify the law of the zero-boundary component $\widetilde{u}_{B,e}$ with the Whittle--Mat\'ern bridge process in Definition~\ref{def:WMB}. 
By \citet[Lemma~2]{split1}, the Cameron--Martin space of the Whittle--Mat\'ern bridge process is
$(H_0^\alpha(e),(\cdot,\cdot)_{\alpha,e})$. 
The zero-boundary component in the edge representation has Cameron--Martin space
$\mathcal{H}_{0,I}(e)$ in the notation of \citet{BSW_Markov}. 
By \citet[Theorem~8]{BSW_Markov},
    $\mathcal{H}_{0,I}(e)\cong \dot{H}_0^\alpha(e)$,
and, by \citet[Lemma~3]{split1},
    $(\dot{H}_0^\alpha(e),(\cdot,\cdot)_{\alpha,e})
    \cong
    (H_0^\alpha(e),(\cdot,\cdot)_{\alpha,e})$.
These two identifications are the Hilbert-space equivalences needed to identify the zero-boundary term in the edge representation with the Whittle--Mat\'ern bridge process of Definition~\ref{def:WMB}. 
Consequently we may write $\widetilde{u}_{B,e}=u_{B,e}$, where $u_{B,e}$ is the Whittle--Mat\'ern bridge on $[0,\ell_e]$. 
Substituting this identification and the expression for $\mv{S}_e(t)$ into the edge representation gives
    $u_e(t)=u_{B,e}(t)+\mv{S}_e(t)B^\alpha u_e$ for 
    $t\in[0,\ell_e]$,
with $u_{B,e}$ independent of $B^\alpha u_e$, as claimed.
\end{proof}

\begin{proof}[Proof of Corollary~\ref{cor:bridge_representation}]
	Consider the Whittle--Mat\'ern field $u(s)$ on  $\Gamma$. 
	By \citet[Theorem~5 and Remark 3]{BSW_Markov}, $u(s_1)$ and $u(s_2)$ are conditionally independent given $\mv{U}$ 
	if $s_1$ and $s_2$ are locations on different edges. 
	By Theorem~\ref{thm:ReprTheoremEdge_Refined}, 
	for any $s=(e_s,t_s) \in \Gamma$
	we can express $u(s)|_{e_s}$ as 
	$u(s)  = \mv{S}_{e_s}(t_s) B^\alpha u_{e_s} + u_{B,e_s}(t_s)$,
	where $u_{B,e_s}$ is a Whittle--Mat\'ern bridge process on $[0, \ell_{e_s}]$, independent of $\mv{U}$. Thus, by defining the low-rank process $u_\Gamma(s) =  \mv{S}_{e}(t) \mv{D}_e \mv{U}$, for $s=(e,t)\in\Gamma$,
	we arrive at the representation \eqref{eq:bridge_representation}.
\end{proof}

\subsection{Proofs for Section~\ref{sec:likandpred} and Section~\ref{sec:predict}}\label{app:proofs:likelihood-prediction}

\begin{proof}[Proof of Proposition~\ref{Them:piAXsoft}]
	Recall that $\mv{U}^* = \mv{TU}$. We have $\pi_{\mv{Y} | \mv{KU}}\left(  \mv{y} | \mv{b} \right)= \pi_{\mv{Y} | \mv{U}_\A^*}(\mv{y}|\mv{b}^*)$ and
			\begin{align}
					\pi_{\mv{Y} | \mv{U}^*_{\A}}\left(  \mv{y} | \mv{b}^* \right)
					&=  \int  \pi_{\mv{Y}|(\mv{U}^*_{\Ac}, \mv{U}^*_{\A})}(\mv{y}|(\mv{u}_{\Ac}^*, \mv{b}^*)) \pi_{\mv{U}^*_{\Ac}| \mv{U}^*_{\A}}(\mv{u}_{\Ac}^*| \mv{b}^*) d\mv{u}^*_{\Ac}.	\label{eq:decomposePI}
			\end{align}
			The goal is to derive an explicit form of the density by evaluating the integral in \eqref{eq:decomposePI}. 
			To shorten the notation, we let $q(\mv{K},\mv{x})$ denote the quadratic form $\mv{x}^\top\mv{K}\mv{x}$ for a matrix $\mv{K}$ and a vector $\mv{x}$. 
			Recall that $\mv{Y}|\mv{U}=\mv{u} \sim \pN(\mv{B}\mv{u}, \mv{\Sigma})$, so that
			\begin{align}\label{eq:proof_cond1}
					\pi_{\mv{Y}|(\mv{U}_\Ac^*,\mv{U}_{\A}^*)}(\mv{y}|(\mv{u}_{\Ac}^*,\mv{b}^*)) 
					&=  \frac{1}{(2\pi)^{\frac{n}{2}}|\mv{\Sigma}|^{1/2}}\exp \left( -\frac{1}{2}q\left(\mv{\Sigma}^{-1}, \mv{y}- \mv{B}^*
					\left[\mv{b}^*, \,\, 
							\mv{u}^*_{\Ac} 
					\right]^{\trsp}
					\right)\right) \\
					&\propto \exp\left(-\frac{1}{2}\mv{u}^{*\top}_{\Ac}\mv{B}^{*\top}_{\Ac,} \mv{\Sigma}^{-1} \mv{B}^*_{,\Ac}\mv{u}^*_{\Ac}+ \mv{y}^{\top}\mv{\Sigma}^{-1}\mv{B}^*_{,\Ac}\mv{u}^*_{\Ac} \right),\nonumber
			\end{align}
			where the proportionality holds as a function of $\mv{u}_{\Ac}^*$.
			By the proof of \citet[][Theorem 2]{bolin2021efficient}, we have
			\begin{equation}\label{eq:Xconditional}
					\mv{U}_{\Ac}^*|\mv{U}^*_\A= \mv{b}^* \sim  \pN\left(
					\mv{\mu}^*_{\Ac}-\left(\mv{Q}_{\Ac\Ac}^{*}\right)^{-1} \mv{Q}_{\Ac\A}^{*} \left( \mv{b}^* -  \mv{\mu}^*_{\A}\right)
					, (\mv{Q}_{\Ac\Ac}^*)^{-1} \right).
			\end{equation}
			By these expressions, 
			\begin{align*}
					&\pi_{\mv{Y}|(\mv{U}^*_{\Ac}, \mv{U}^*_{\A})}(\mv{y}|(\mv{u}_{\Ac}^*, \mv{b}^*)) \pi_{\mv{U}^*_{\Ac}| \mv{U}^*_{\A}}(\mv{u}_{\Ac}^*| \mv{b}^*) =\\
					&= \exp\left(-\frac{1}{2}\mv{u}^{*\top}_{\Ac} \mv{B}^{*\top}_{\Ac,}\mv{\Sigma}^{-1} \mv{B}^*_{,\Ac} \mv{u}^*_{\Ac}+ \left( \mv{B}^{*\top}_{\Ac,} \mv{\Sigma}^{-1}\mv{y} \right)^{\top}\mv{u}^*_{\Ac} \right) \frac{ |\mv{Q}^*_{\Ac\Ac}|^{1/2}}{\left(2\pi\right)^{\nicefrac{(n+m-k)}{2}} |\mv{\Sigma}|^{1/2}} \cdot \\
					&\quad\exp\left(-\frac{1}{2} \mv{u}^{*\top}_{\Ac} \mv{Q}^*_{\Ac\Ac} \mv{u}^*_{\Ac}+ \left(  \mv{Q}^*_{\Ac\Ac} \widetilde{\mv{\mu}}^*_{\Ac}\right)^{\top} \mv{u}^*_{\Ac}   \right)
					\exp \left(- \frac{1}{2}\left[  \mv{y}^{\top}\mv{\Sigma}^{-1}\mv{y}+  \widetilde{\mv{\mu}}_{\Ac}^{*\top}  \mv{Q}^*_{\Ac\Ac}  \widetilde{\mv{\mu}}^*_{\Ac} \right] \right) \\
					&= 
					\frac{\pi_{\mv{U}^*_{\Ac}|(\mv{Y}, \mv{U}^*_{\A})}(\mv{u}_{\Ac}^*| (\mv{y},\mv{b}^*)) \exp \left( \frac{1}{2}  \widehat{\mv{\mu}}_{\Ac}^{*\top} \widehat{\mv{Q}}^*_{\Ac\Ac} \widehat{\mv{\mu}}^*_{\Ac} \right) }{|\mv{Q}^*_{\Ac\Ac}|^{-1/2}|\mv{\Sigma}|^{1/2} \left(2\pi\right)^{\nicefrac{n}{2}} }
					\exp \bigl(- \frac{1}{2}\bigl[  \mv{y}^{\top} \mv{\Sigma}^{-1}\mv{y} +  \widetilde{\mv{\mu}}_{\Ac}^{*\top}  \mv{Q}^*_{\Ac\Ac}\widetilde{\mv{\mu}}^*_{\Ac}\bigr]\bigr).
			\end{align*}
			The desired result is now obtained by 
			inserting this expression in \eqref{eq:decomposePI} and evaluating the integral, noting  that $\pi_{\mv{U}^*_{\Ac}|(\mv{Y}, \mv{U}^*_{\A})}(\mv{u}_{\Ac}^*| (\mv{y},\mv{b}^*))$ integrates to 1. 
\end{proof}

\begin{proof}[Proof of Lemma~\ref{lem:kriging_alpha1}]
Using the bridge representation from Corollary~\ref{cor:bridge_representation}, we have
	\begin{align} 
	\mathbb{E}\left[u(s^\star) \mid \mv{Y}\right] =& \mathbb{E}\left[ u_\Gamma(s^\star)  +u_{B,e^\star}(t^\star) \mid \mv{Y} \right] 
=   \mathbb{E}[u_\Gamma(s^\star)\mid \mv{Y}] + \mathbb{E}\left[ u_{B,e^\star}(t^\star) \mid \mv{Y} \right] \nonumber \\
=&  \mathbb{E}[u_\Gamma(s^\star)\mid \mv{Y}]+  \mathbb{E}\left[\mathbb{E}\left[ u_{B,e^\star}(t^\star) \mid \{\mv{Y}_{e^\star},{u}_\Gamma(\mv{s}_{e^\star})\}  \right]\mid \mv{Y}\right]\label{eq:kriging3}  \\
=&  \mathbb{E}[u_\Gamma(s^\star)\mid \mv{Y}] +  \mathbb{E}\left[\mathbb{E}\left[ u_{B,e^\star}(t^\star) \mid \mv{Y}_{e^\star}-{u}_\Gamma(\mv{s}_{e^\star})  \right]\mid \mv{Y} \right]\label{eq:kriging4}  \\
=& \mathbb{E}[u_\Gamma(s^\star)\mid \mv{Y}] +  \mathbb{E}\left[r_{e^{\star}}(t^\star, \mv{t}_{e^\star})\mv{\Sigma}_{e^\star}^{-1}(\mv{Y}_{e^\star} - {u}_\Gamma(\mv{s}_{e^\star}))\mid \mv{Y} \right]\label{eq:kriging5}\\
=& \mathbb{E}[u_\Gamma(s^\star)\mid \mv{Y}] +r_{e^{\star}}(t^\star, \mv{t}_{e^\star})\mv{\Sigma}_{e^\star}^{-1}\left(\mv{Y}_{e^\star} - \mathbb{E}\left[{u}_\Gamma(\mv{s}_{e^\star})\mid \mv{Y} \right]\right),\nonumber
\end{align}
where in \eqref{eq:kriging3} we used the tower property of conditional expectations, in \eqref{eq:kriging4} 
we used the fact that $\sigma(\{\mv{Y}_e, u_\Gamma(\mv{s}_e)\}_{e\in\mathcal{E}})=\sigma(\{\mv{Y}_e - u_\Gamma(\mv{s}_e)\}_{e\in\mathcal{E}}, \{u_\Gamma(\mv{s}_e)\}_{e\in\mathcal{E}})$,
 that $\{u_\Gamma(\mv{s}_e)\}_{e\in\mathcal{E}}$ is independent of $u_{B,e}$, 
and that $\mv{Y}_{\widetilde{e}} - u_\Gamma(\mv{s}_{\widetilde{e}})= u_{B,\widetilde{e}}+\mv{\varepsilon}_{\widetilde{e}}$ is independent of
$u_{B,e}$ for $e\neq \widetilde{e}$. In \eqref{eq:kriging5}, we used that $\mathbb{E}\left[ u_{B,e}(t^\star) \mid \mv{Y}_e-{u}_\Gamma(\mv{s}_e)  \right] = \mathbb{E}\left[ u_{B,e}(t^\star) \mid {u}_{B,e}(\mv{t}_e) + \mv{\varepsilon}_e \right]$, 
applied the kriging formula, and used $u_{B,\widetilde{e}}+\mv{\varepsilon}_{\widetilde{e}}=\mv{Y}_{\widetilde{e}} - u_\Gamma(\mv{s}_{\widetilde{e}})$.
\end{proof}

\begin{proof}[Proof of Proposition~\ref{Them:piXgby}]
			Begin by noting that
			$$
			\pi_{\mv{U}_{\Ac}^*|(\mv{Y},\mv{U}^*_{\A})}(\mv{u}^*_{\Ac}| (\mv{y},\mv{b}^*)) \propto
			\pi_{\mv{Y}|(\mv{U}_\A^*,\mv{U}_{\Ac}^*)}(\mv{y}|(\mv{b}^*, \mv{u}_{\Ac}^*)) \pi(\mv{u}^*_{\Ac}|\mv{b}^*).
			$$ 
			By \eqref{eq:proof_cond1} and \eqref{eq:Xconditional} we have  
			$\pi_{\mv{U}^*_{\Ac}|\mv{U}^*_\A}(\mv{u}^*_{\Ac}|\mv{b}^*) \propto
			\exp \left( -\frac{1}{2}  \left( \mv{u}^*_{\Ac}- \widetilde{\mv{\mu}}^*_{\Ac} \right)^\top  \mv{Q}^*_{\Ac\Ac}  \left( \mv{u}^*_{\Ac}- \widetilde{\mv{\mu}}^*_{\Ac} \right) \right),$ 
			as a function of $\mv{u}_{\Ac}^*$,
			where $\widetilde{\mv{\mu}}^*_{\Ac} =  \mv{\mu}^*_{\Ac}-\left(\mv{Q}_{\Ac\Ac}^{*} \right)^{-1}\mv{Q}_{\Ac\A}^{*} \left( \mv{b}^* -  \mv{\mu}^*_{\A}\right)$.
			Thus, 
			\begin{align*}
					\pi_{\mv{U}_{\Ac}^*|(\mv{Y},\mv{U}^*_{\A}) } (\mv{u}^*_{\Ac}|(\mv{y},\mv{b}^*))  \propto
					&\exp\left(-\frac{1}{2}\mv{u}^{*\top}_{\Ac} \mv{B}^{*\top}_{\Ac,}\mv{\Sigma}^{-1} \mv{B}^*_{,\Ac}\mv{u}^*_{\Ac}+ \left( \mv{B}^{*\top}_{\Ac,} \mv{\Sigma}^{-1}\mv{y} \right)^{\top}\mv{u}^*_{\Ac} \right) \times \\  
					\times& \exp\left(-\frac{1}{2} \mv{u}^{*\top}_{\Ac} \mv{Q}^*_{\Ac\Ac} \mv{u}^*_{\Ac}+ \left(  \mv{Q}^*_{\Ac\Ac} \widetilde{\mv{\mu}}^*_{\Ac}\right)^{\top} \mv{u}^*_{\Ac}   \right) \\
					\propto 	&\exp\left(- \frac{1}{2} \left(\mv{u}_{\Ac}^*-\widehat{\mv{\mu}}_{\Ac}^* \right)^\top \widehat{\mv{Q}}_{\Ac\Ac}^*
					\left(\mv{u}_{\Ac}^*-\widehat{\mv{\mu}}_{\Ac}^* \right)\right).
			\end{align*}
			Finally, the relation $\mv{U}= \mv{T}^\top\mv{U}^*$ completes the proof.
\end{proof}

\begin{proof}[Proof of Proposition~\ref{prop:kriging_corrected_simulation}]
Since $\widetilde{\mv{x}}_e(\mv{b}_e)$ is a linear transformation of a Gaussian
vector, it is Gaussian. Moreover, $\pE\{\widetilde{\mv{x}}_e(\mv{b}_e)\}
    =
    \mv{S}_e(\mv{t}_e^\star)\mv{b}_e$ and 
\begin{align*}
    \Cov\{\widetilde{\mv{x}}_e(\mv{b}_e)\}
    &=
    \Cov\{
        \mv{x}_{e,o}
        -
        \mv{S}_e(\mv{t}_e^\star)\mv{x}_{e,\partial}
    \} 
    =
    \mv{\Sigma}_{M,e}^\star
    -
    \mv{S}_e(\mv{t}_e^\star)
    \mv{\Sigma}_{e,\partial\partial}
    \mv{S}_e(\mv{t}_e^\star)^\trsp 
    =
    \mv{\Sigma}_e^\star .
\end{align*}
This proves the result.
\end{proof}

\bibliography{unified_graph_bib}

\end{document}